\documentclass[a4paper,USenglish,cleveref,nameinlink, autoref, thm-restate, numberwithinsect]{lipics-v2021}
\usepackage{graphicx} 
\usepackage{amsthm}
\usepackage{amssymb}
\usepackage{amsmath}
\usepackage{booktabs}
\usepackage{enumerate}
\usepackage{bm}
\usepackage{multirow}

\usepackage{makecell}

\theoremstyle{definition}
\newtheorem{problem}{Problem}

\newcommand{\cf}{\ensuremath{\mathcal{F}}}
\newcommand{\cl}{\ensuremath{\mathcal{L}}}
\newcommand{\B}{\ensuremath{\mathcal{B}}}
\newcommand{\A}{\ensuremath{\mathcal{A}}}

\newcommand{\NP}{\ensuremath{\mathsf{NP}}}
\newcommand{\coNP}{\ensuremath{\mathsf{coNP}}}
\newcommand{\poly}{\ensuremath{\mathsf{poly}}}
\newcommand{\XP}{\ensuremath{\mathsf{XP}}}
\newcommand{\FPT}{\ensuremath{\mathsf{FPT}}}
\newcommand{\W}{\ensuremath{\mathsf{W[1]}}}
\newcommand{\WT}{\ensuremath{\mathsf{W[2]}}}
\renewcommand{\O}{\ensuremath{\mathcal{O}}}
\newcommand{\N}{\ensuremath{\mathbb{N}}}
\newcommand{\nZ}{\ensuremath{\overline{Z}}}

\newcommand{\ThreeColor}{\textsc{3-Coloring}}

\newcommand{\MinLeafL}{\textsc{Min-Leaf \cl-Tree}}
\newcommand{\MaxLeafL}{\textsc{Max-Leaf \cl-Tree}}

\newcommand{\MinInL}{\textsc{Min-Internal \cl-Tree}}
\newcommand{\MaxInL}{\textsc{Max-Internal \cl-Tree}}

\newcommand{\MinLeaf}{\textsc{Min-Leaf}}
\newcommand{\MaxLeaf}{\textsc{Max-Leaf}}

\newcommand{\MinIn}{\textsc{Min-Internal}}
\newcommand{\MaxIn}{\textsc{Max-Internal}}

\usepackage{tikz}
\usetikzlibrary{arrows}
\usetikzlibrary{shapes.geometric}
\usetikzlibrary{decorations.pathreplacing}
\usetikzlibrary{shadows.blur}
\usetikzlibrary{shapes.symbols}

\usetikzlibrary{arrows}
\usetikzlibrary{math}
\usetikzlibrary{calc}
\usetikzlibrary{fadings}
\tikzstyle{treeedge}=[line width=2pt, lipicsYellow]
\tikzstyle{blobedge}=[ultra thick, lipicsLineGray]
\tikzstyle{cliqueedge}=[line width=3pt, lipicsYellow]
\tikzstyle{normaledge}=[black,thick]
\tikzstyle{fadeedge}=[black, thick,path fading=fade ends]
\tikzstyle{noneedge}=[line width=3pt, dashed, lipicsLineGray]
\tikzstyle{exampleedge}=[line width=3pt, lipicsLineGray]

\tikzstyle{thinedge}=[ultra thick, lipicsLineGray,opacity=0.1]
\tikzset{vertex/.style = {draw, circle, thick, fill=white,minimum size = 5pt,inner sep=0pt}}
\title{On Kernels and Leaves: Searching for Bare and Lush Trees}

\author{Jesse Beisegel}{Institute of Mathematics, Brandenburg University of Technology, Cottbus, Germany}{jesse.beisegel@b-tu.de}{https://orcid.org/0000-0002-8760-0169}{}
\author{Ekkehard Köhler}{Institute of Mathematics, Brandenburg University of Technology, Cottbus, Germany}{ekkehard.koehler@b-tu.de}{}{}
\author{Robert Scheffler}{Institute of Mathematics, Brandenburg University of Technology, Cottbus, Germany}{robert.scheffler@b-tu.de}{https://orcid.org/0000-0001-6007-4202}{}
\author{Martin Strehler}{Department of Mathematics, Westsächsische Hochschule Zwickau, Zwickau, Germany}{martin.strehler@whz.de}{https://orcid.org/0000-0003-4241-6584}{}

\Copyright{Jesse Beisegel, Ekkehard Köhler, Robert Scheffler, Martin Strehler}

\keywords{graph search, spanning tree, parameterized complexity, kernelization, leaves of trees}
\authorrunning{J. Beisegel, E. Köhler, R. Scheffler, and M. Strehler}

\ccsdesc[500]{Mathematics of computing~Graph algorithms} 

\ccsdesc[500]{Theory of computation~Parameterized complexity and exact algorithms}

\hideLIPIcs
\nolinenumbers

\begin{document}

\maketitle

\begin{abstract}
We study a variation of the classical Maximum (Minimum) Leaf Spanning Tree problem. In many applications, Depth-First Search (DFS) is used to compute a spanning tree of a graph. Such a search tree is constructed by connecting each vertex $v$ with the last vertex the search has visited before $v$ and we call this a \emph{last-in tree}. By restricting the Maximum (Minimum) Leaf Spanning Tree problem to last-in trees of a graph search, we ask for a search ordering that leads to the largest (smallest) number of leaves in its search tree. Recently, Bergougnoux et al. [\emph{Journal of Computer and System Sciences} 154 (2025)] have studied the parameterized complexity of these problems for DFS. They showed that the minimization problem is para-\NP-hard  and the maximization problem is \W-hard when parameterized by the number of leaves. When parameterized by the number of internal vertices, both problems have polynomial kernels. Here, we examine whether these results also hold for the variant Lexicographic DFS (LDFS). We show that the hardness results of DFS can be transferred to LDFS. We also present exponential kernels for the number of internal vertices as the parameter. We complement this by showing that polynomial kernels do not exist, unless $\NP \subseteq \coNP / \poly$. We also consider last-in trees of searches that do not follow the DFS scheme. In contrast to (L)DFS, minimizing the number of internal vertices is para-\NP-hard for several searches including Breadth-First Search. 
\end{abstract}

\section{Introduction}

Graph searches, despite being relatively simple procedures, serve as crucial sub-routines in many high-level algorithms. Depth-First Search (DFS), for example, helps to recognize planar graphs~\cite{hopcraft1974planarity}, find strongly connected components~\cite{tarjan1972dfs}, or construct a topological sorting for directed acyclic graphs~\cite{cormen2022introduction}. Consequently, the properties of these searches and their associated search trees have been studied extensively, in particular the question of whether a given spanning tree is a search tree~\cite{beisegel2021recognition,mfcs,hagerup1985biconnected,hagerup1985recognition,korach1989dfs,manber1990recognizing,scheffler2022recognition}.

Recently, Bergougnoux et al.~\cite{bergougnoux2025parameterized} addressed the problem of finding DFS trees with a minimum or maximum number of leaves for a given graph. The authors study the parametrized complexity for both minimization and maximization when either the number of leaves or the number of internal vertices is used as parameter. In particular, they show that deciding whether there is a DFS tree with at most $k$ leaves is para-\NP-hard, while deciding whether a DFS tree can have at least $k$ leaves is \W-hard. In contrast, the problems of deciding whether a DFS tree can have at most or at least $k$ internal vertices admit polynomial kernels and thus are fixed-parameter tractable (\FPT). 

Historically, the study of spanning trees with either few or many leaves is a much older problem. For example, a spanning tree with one leaf corresponds to a Hamiltonian path.\footnote{Following the convention in~\cite{bergougnoux2025parameterized}, we only consider rooted spanning trees and do not count the root as a leaf even if it has degree~1.} Therefore, the problem is already hard for the fixed value $k = 1$. However, when we consider the dual parameterization by the number of internal vertices, it allows for linear kernels and, thus, becomes \FPT~\cite{fomin2013linear,li2017deeper}. That is why the problem is commonly referred to as the \textsc{Maximum Internal Spanning Tree} (MIST) problem. Conversely, for the \textsc{Maximum Leaf Spanning Tree} (MLST) problem, known to be \NP-hard even for restricted graph classes such as planar cubic graphs~\cite{reich16complexity}, a large number of \FPT{} algorithms and linear kernels have been presented when parameterized by the number of leaves~\cite{bodlaender1993linear,bonsma2003faster,bonsma2008spanning,estivill-castro2005fpt,fellows1992well,fellows2000coordinatized,kneis2011new,raible2010amortized,zehavi2018k-leaf}. However, when parameterized by the number of internal vertices, it is equivalent to the \textsc{Connected Dominating Set} problem~\cite{fujie2003exact} and, thus, \WT-complete~\cite{cygan2015param}.

Comparing the results on general spanning trees with the results of Bergougnoux~et~al.~\cite{bergougnoux2025parameterized} on DFS trees, we can observe that for both pairs of dual parameterizations, there is one case that is in \FPT{} while the other case is hard. For general spanning trees, both parameterization by the number of leaves and by the number of internal vertices are hard in one case and tractable in the other. In contrast, for DFS only the parameterizations by the internal vertices are tractable, whereas the problems are hard when parameterized by the number of leaves. 

Most recently, Beisegel et al.~\cite{iwoca} addressed the same questions for BFS trees. They found that the complexity of the problems is somehow inverted compared to DFS. For BFS trees, both minimizing and maximizing the number of leaves is in \FPT{} when parameterized by the number of leaves, but \W-hard when parameterized by the number of internal vertices.

These differences may arise from the different schemes in which BFS and DFS trees are defined. In the terminology of~\cite{beisegel2021recognition}, DFS trees are \emph{last-in trees} or \emph{\cl-trees}, while BFS trees are \emph{first-in trees} or \emph{\cf-trees}. Given a search ordering $\sigma=(v_1,v_2,\dots,v_n)$, i.e., the sequence of visits of the vertices during a graph search, the \cl-tree is constructed by connecting each vertex $v_i$ with $i > 1$ to the rightmost adjacent predecessor in the sequence. The \cf-tree, however, is constructed by connecting $v_i$ to the leftmost adjacent predecessor (see \cref{fig:trees}). 

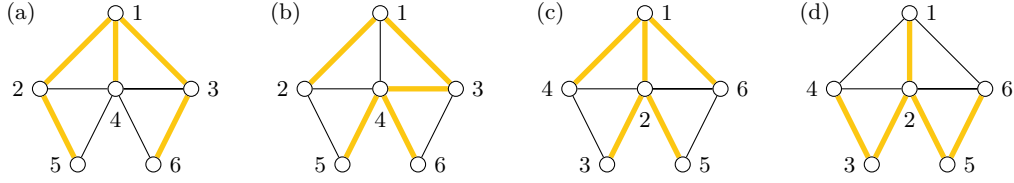
\begin{figure}
	\centering
	\begin{tikzpicture}[vertex/.style={inner sep=2pt,draw,circle}]
    \footnotesize
    \begin{scope}
	\node[vertex, label={180:$2$}] (1) at (-0.25,0) {};
	\node[vertex, label={[label distance=0.1cm]-90:$4$}] (2) at (0.75,0) {};
	\node[vertex, label={0:$3$}] (3) at (1.75,0) {};
	\node[vertex, label={0:$1$}] (4) at (0.75,1) {};
	\node[vertex, label={180:$5$}] (5) at (0.25,-1) {};
	\node[vertex, label={0:$6$}] (6) at (1.25,-1) {};
	\node[] (a) at (-0.5,1) {(a)};
	
	\draw[] (1)--(2)--(3)--(4)--(2)--(6)--(3)--(2)--(5)--(1)--(4);
	\draw[treeedge] (5)--(1)--(4)--(3)--(6);
	\draw[treeedge] (2) -- (4);
	\end{scope}
	
	\begin{scope}[xshift=3.5cm]
	\node[vertex, label={180:$2$}] (1) at (-0.25,0) {};
	\node[vertex, label={[label distance=0.1cm]-90:$4$}] (2) at (0.75,0) {};
	\node[vertex, label={0:$3$}] (3) at (1.75,0) {};
	\node[vertex, label={0:$1$}] (4) at (0.75,1) {};
	\node[vertex, label={180:$5$}] (5) at (0.25,-1) {};
	\node[vertex, label={0:$6$}] (6) at (1.25,-1) {};
	\node[] (b) at (-0.5,1) {(b)};
	
	\draw[] (1)--(2)--(3)--(4)--(2)--(6)--(3)--(2)--(5)--(1)--(4);
	\draw[treeedge] (5)--(2)--(3)--(4)--(1);
	\draw[treeedge] (2) -- (6);
	\end{scope}
	
	\begin{scope}[xshift=7cm]
	\node[vertex, label={180:$4$}] (1) at (-0.25,0) {};
	\node[vertex, label={[label distance=0.1cm]-90:$2$}] (2) at (0.75,0) {};
	\node[vertex, label={0:$6$}] (3) at (1.75,0) {};
	\node[vertex, label={0:$1$}] (4) at (0.75,1) {};
	\node[vertex, label={180:$3$}] (5) at (0.25,-1) {};
	\node[vertex, label={0:$5$}] (6) at (1.25,-1) {};
	\node[] (a) at (-0.5,1) {(c)};
	
	\draw[] (1)--(2)--(3)--(4)--(2)--(6)--(3)--(2)--(5)--(1)--(4);
	\draw[treeedge] (1)--(4)--(2)--(5);
	\draw[treeedge] (2) -- (6);
	\draw[treeedge] (3) -- (4);
	\end{scope}
	
	\begin{scope}[xshift=10.5cm]
	\node[vertex, label={180:$4$}] (1) at (-0.25,0) {};
	\node[vertex, label={[label distance=0.1cm]-90:$2$}] (2) at (0.75,0) {};
	\node[vertex, label={0:$6$}] (3) at (1.75,0) {};
	\node[vertex, label={0:$1$}] (4) at (0.75,1) {};
	\node[vertex, label={180:$3$}] (5) at (0.25,-1) {};
	\node[vertex, label={0:$5$}] (6) at (1.25,-1) {};
	\node[] (a) at (-0.5,1) {(d)};
	
	\draw[] (1)--(2)--(3)--(4)--(2)--(6)--(3)--(2)--(5)--(1)--(4);
	\draw[treeedge] (1)--(5)--(2)--(6)--(3);
	\draw[treeedge] (2) -- (4);
	\end{scope}

	\end{tikzpicture}\caption{Illustration of \cf-trees and \cl-trees borrowed from \cite{scheffler2022recognition}. The search trees are denoted by the thick yellow edges. The tree in (a) is the \cf-tree of the given BFS ordering. The $\cl$-tree of that ordering is given in (b). The graphs in (c) and (d) present a DFS ordering together with its \cf-tree~(c) and \cl-tree (d).}\label{fig:trees}
\end{figure}

The hardness results for BFS given in \cite{iwoca} do not only hold for BFS trees but for \cf-trees of a large family of graph searches including DFS. This motivates several questions. First of all, one may ask whether the hardness results for the number of leaves as parameter given for \cl-trees of DFS can be extended to the \cl-trees of other searches. Conversely, one may ask whether one can also give \FPT{} algorithms and polynomial kernels when parameterized by the number of internal vertices. In particular, this question is interesting for \emph{Lexicographic DFS} (LDFS). This search, introduced by Corneil and Krueger~\cite{corneil2008unified}, has gained significant attention in the context of optimization problems on cocomparability graphs~\cite{corneil2013ldfs,corneil2016power,mertzios2012simple,mertzios2018linear}. In essence, it is a DFS with an additional tie-breaking scheme where not only the most recently visited neighbor of a vertex is important, but also its entire history of visited neighbors. Of course, all \cl-trees of LDFS are also \cl-trees of DFS, and thus they have the same strong structure. However, the vertex choice scheme of LDFS is more enhanced than that of DFS which may influence the complexity of finding search trees with few or many leaves. The study of these problems for LDFS might give insights as to whether the tractability of the problems on DFS depends  only on the structure of the trees or also on how the respective search determines the vertex order.

\subparagraph*{Our Contribution.} We study the problem of finding $\cl$-trees with either few or many leaves. For a given graph $G$ and a search paradigm $\A$, we explore both the minimizing and maximizing variants of the following two problems, which vary depending on their use of the parameter 
$k$. The first problem is \textsc{Min-Leaf (Max-Leaf) \cl-Tree} for a search $\A$: given a graph $G$ and an integer $k$, we ask whether there exists an $\A$-ordering of $G$ whose $\cl$-tree has at most (least) $k$ leaves. The second problem is \textsc{Min-Internal (Max-Internal) \cl-Tree} for a search $\A$: given a graph $G$ and a parameter $k$, we ask for an $\A$-ordering of $G$ whose $\cl$-tree has at most (least) $k$ internal vertices. Note that these problems differ mainly in how the parameter $k$ is defined. Without parameterization, minimizing the number of leaves is the same as maximizing the number of internal vertices and vice versa.%

In particular, we focus on Lexicographic Depth-First Search (LDFS), Breadth-First Search (BFS), and Generic Search (GS). The latter is the most general graph search considered here, where the only condition on the next vertex to be visited is that it has a visited neighbor. In \cref{sec:internal}, we show that we can extend the \FPT{} results for the parameter ``number of internal vertices'' of DFS to LDFS by presenting exponential kernels. In contrast to DFS, we show that there are no polynomial kernels, unless $\NP \subseteq \coNP / \poly$. We also show in this section that minimizing the number of leaves of \cl-trees is para-\NP-hard when considered for GS or BFS. In \cref{sec:leaves}, we extend the hardness results for DFS with respect to the parameter ``number of leaves'' to LDFS. To this end, we present a graph modification called \emph{blow-up} that ensures that DFS $\cl$-trees with $k$ leaves of the original graph correspond to LDFS \cl-trees of the blow-up with $k$ leaves. A summary of our results is given in \cref{tab:results-tree}.

\begin{table}[t]
	\centering
	\caption{Known and new results about the parameterized complexity of finding \cl-trees with minimal or maximal number of leaves. The non-existence of polynomial kernels holds under the assumption that $\NP \not\subseteq \coNP / \poly$.
}\label{tab:results-tree}
    \setlength{\tabcolsep}{11pt}
	\begin{tabular}{l c c c c }
		\addlinespace
		\toprule
    
		\footnotesize results   &\footnotesize GS &\footnotesize BFS                           &\footnotesize DFS~\cite{bergougnoux2025parameterized}  &\footnotesize LDFS                 \\[0.8ex] \midrule
        \footnotesize\MinLeaf{}     &\footnotesize 
		  para-NP-hard  &\footnotesize 
		para-NP-hard  &\footnotesize 
		para-NP-hard & 
		para-NP-hard  
		         \\[2ex]
		\footnotesize{\MaxLeaf{}} &\footnotesize 
		?  &\footnotesize 
		?  &\footnotesize 
		W[1]-hard& 
		W[1]-hard 
        \\[1.2ex] 
        \footnotesize \MaxIn{}     &\footnotesize 
		  $\O(k)$ kernel  &\footnotesize 
		  ?   &\footnotesize 
		$\O(k^3)$ kernel& 
		\makecell{$\O(k \cdot 4^k)$ kernel \\ no poly kernel} 
        \\
		\footnotesize\MinIn{} &\footnotesize 
		  para-NP-hard   &\footnotesize 
		para-NP-hard   &\footnotesize 
		$\O(k^3)$ kernel & 
		\makecell{$\O(k \cdot 4^k)$ kernel\\ no poly kernel} 
        \\[0.8ex] 
		\bottomrule \addlinespace
	\end{tabular}
\end{table}

\section{Preliminaries}

Throughout this paper, all graphs are simple, undirected, connected, and non-empty. We denote the set of \emph{vertices} of a graph $G$ by $V(G)$ and the set of \emph{edges} by $E(G)$, where $n:=|V(G)|$ and $m:=|E(G)|$. Furthermore, we use $N_G(v)$ to denote the \emph{neighborhood} of a vertex $v\in V(G)$. For further graph-theoretic concepts, we refer to~\cite{west2001introduction}. For a comprehensive overview of the concepts of parameterized complexity,  we refer the reader to \cite{cygan2015param,downey2013fundamentals}. 

A \emph{graph search} $\cal A$ is a procedure to systematically visit all vertices of a graph. The associated \emph{search ordering} $\sigma=(v_1,\dots,v_n)$ represents the order in which the vertices are visited. Note that we write $u \prec_\sigma v$ if the vertex $u$ appears before $v$ (also called ``to the left of $v$'' in the following) in the search ordering $\sigma$. Here, we consider the following searches~$\cal A$. \emph{Generic Search} (GS) imposes no constraints beyond connectivity, i.e., after the first vertex, each newly visited vertex must simply be a neighbor of an already visited vertex. BFS refers to the standard \emph{Breadth-First Search} implemented through a queue data structure~\cite{Kleinberg06algorithmdesign}. \emph{Depth-First Search} (DFS) proceeds by exploring an adjacent vertex to the most recently visited vertex, backtracking when no unvisited neighbor remains. This traversal can be implemented using a stack data structure~\cite{Kleinberg06algorithmdesign}. \emph{Lexicographic Depth-First Search} (LDFS) employs lexicographically ordered labels to resolve ties when the most recently visited vertex has multiple unvisited neighbors~\cite{corneil2008unified}. Alternatively, LDFS can be defined using the \emph{4-point condition} given in~\cite{corneil2008unified}.

\begin{lemma}[Corneil and Krueger~\cite{corneil2008unified}]\label{lemma:4point-ldfs}
    The vertex ordering $\sigma$ is an LDFS ordering of some graph $G$ if and only if for all $a,b,c \in V(G)$ with $a \prec_\sigma b \prec_\sigma c$, $ac \in E(G)$ and $ab \notin E(G)$, there is a vertex $d \in V(G)$ with $a \prec_\sigma d \prec_\sigma b$ such that $db \in E(G)$ and $dc \notin E(G)$.
\end{lemma}

This lemma can also be formulated in the following way: If set $A$ contains the neighbors of a vertex $c$ among the $k$ most recently visited vertices by LDFS and another vertex $b$ is not adjacent to at least one $a \in A$ and to no vertex outside of $A$ among these last $k$ vertices, then $b$ cannot be visited next.
A similar characterization can also be given for DFS orderings.

\begin{lemma}[Corneil and Krueger~\cite{corneil2008unified}]\label{lemma:4point-dfs}
    The vertex ordering $\sigma$ is an DFS ordering of some graph $G$ if and only if for all $a,b,c \in V(G)$ with $a \prec_\sigma b \prec_\sigma c$, $ac \in E(G)$ and $ab \notin E(G)$, there is a vertex $d \in V(G)$ with $a \prec_\sigma d \prec_\sigma b$ such that $db \in E(G)$.
\end{lemma}

 Following~\cite{beisegel2021recognition}, a spanning tree $T$ of $G$ is called the \cl-tree of a search ordering $\sigma=(v_1,v_2,\dots,v_n)$ if $T$ is constructed as follows: $v_iv_j\in E(T)$ with $i<j$ if and only if $v_i\in N(v_j)$ and $v_\ell\not\in N(v_j)$ for all $\ell$ with $i<\ell<j$. In other words, in the tree $T$ each vertex $v$ is connected to the neighbor in $G$ that was \emph{last} visited prior to $v$ in $\sigma$.
 The \cl-tree of a search is rooted at the first vertex $r\in V(G)$ of the search ordering. A vertex $v \in V(T)$ is a \emph{leaf} if it has no descendants and is called an \emph{internal vertex} otherwise. Note that we do not consider the root to be a leaf, even if it has degree one. Further note that the $\cl$-tree is well-defined if the search ordering is an ordering that may be produced by GS.

There are various strong properties of $\cl$-trees that can be used algorithmically. For our purposes, the following observation will be helpful.

\begin{observation}\label{obs:dfs-leaves}
    Let $T$ be the $\cl$-tree of a DFS ordering $\sigma$ of the graph $G$. The following three conditions are equivalent for a vertex $v$:
    \begin{enumerate}[(i)]
        \item $v$ is a leaf of $T$,
        \item if $v$ has a successor $w$ in $\sigma$, then $w$ is not a neighbor of $v$,
        \item all neighbors of $v$ are to the left of $v$ in $\sigma$.
    \end{enumerate}
\end{observation}

\section{Number of Internal Vertices as Parameter}\label{sec:internal}
\subsection{LDFS Trees}

Bergougnoux et al.~\cite{bergougnoux2025parameterized} have shown that the two problems \MinInL{} and \MaxInL{} of DFS have polynomial kernels and, therefore, are in \FPT{}. We extend their results  to \cl-trees of LDFS to some extent. First, we show that the structural properties of DFS \cl-trees also yield exponential kernels for both problems when considering the special case of LDFS \cl-trees. However, we also show that the more sophisticated search rule of LDFS appears to  increase the parameterized complexity of the problems. To this end, we show that \MinIn{} and \MaxInL{} of LDFS do not admit polynomial kernels unless $\NP \subseteq \coNP / \poly$.

\subparagraph{Exponential Kernels}

The polynomial kernels of Bergougnoux et al.~\cite{bergougnoux2025parameterized} for the DFS problems hinge on the following results, which also form important ingredients for our algorithms.

\begin{observation}[Bergougnoux et al.~\cite{bergougnoux2025parameterized}]\label{lemma:dfs-vc1}
    The set of internal vertices of any \cl-tree of DFS of a connected graph $G$ is a vertex cover of $G$.
\end{observation}

\begin{lemma}[Bergougnoux et al.~\cite{bergougnoux2025parameterized}]\label{lemma:dfs-vc3}
    Let $G$ be a connected graph and let $Z$ be a vertex cover of $G$. Then, every rooted spanning tree of $G$ has at most $2|Z|$ internal vertices, and at most $|Z|$ internal vertices are not in $Z$.
\end{lemma}
The next result shows that, for a graph $G$, one can add false twins of leaves of an LDFS $\cl$-tree of $G$ without changing the number of internal vertices of the LDFS $\cl$-tree.

\begin{lemma}\label{lemma:adding-twin}
    Let $\sigma$ be an LDFS ordering of a graph $G$ whose $\cl$-tree $T$ has $k$ internal vertices and let $v$ be a leaf in $T$. Furthermore, let $H$ be the graph that is constructed by adding a false twin $w$ of $v$ to $G$. Adding $w$ to the ordering $\sigma$ directly after $v$ creates an LDFS ordering of $H$ whose $\cl$-tree has exactly $k$ internal vertices.
\end{lemma}

\begin{proof}
    First, we show that the new ordering is an LDFS ordering of $G'$. Let $a$, $b$, and $c$ be in $V(G')$ such that $a \prec_{\sigma'} b \prec_{\sigma'} c$, $ac \in E(G)$, and $ab \notin E(G)$. Due to \cref{obs:dfs-leaves}, $a$ is neither $v$ nor $w$. Furthermore, since $v$ and $w$ are false twins, at most one of the two can appear in $\{b,c\}$. Thus, if $w$ equals $b$ or $c$, then we can replace it with $v$ and get another triple of $\sigma$ fulfilling the prerequisites of the 4-point condition given in \cref{lemma:4point-ldfs}. This implies that there is a vertex between $a$ and $b$ that is adjacent to $b$ and not to $c$, and hence $\sigma'$ is an LDFS ordering.

    Due to \cref{obs:dfs-leaves}, $w$ is a leaf in the \cl-tree of $\sigma'$. Furthermore, the rightmost neighbor of $w$ in $\sigma'$ is the rightmost neighbor of $v$ in $\sigma'$. Hence, $v$ and $w$ have the same parent in the \cl-tree of $\sigma'$ and the number of internal vertices did not change.
\end{proof}

In the following, we aim to bound the kernel size of the problems using the vertex cover number. As the internal vertices of a spanning tree form a vertex cover of the underlying graph, we can use this to construct a kernel parametrized by the solution size.

\begin{lemma}\label{lemma:ldfs-kernel-vc}
    \MinInL{} and \MaxInL{} of LDFS admit kernels with $\O(\tau \cdot 4^{\tau})$ vertices when parameterized by the vertex cover number $\tau$ of the graph.
\end{lemma}

\begin{proof}
    We use the folklore greedy 2-approximation algorithm attributed to Gavril as well as Yannakakis in {\cite[p.~432]{papadimitriou1998combinatorial}} to find a vertex cover $Z$ of $G$ of size $s \leq 2\tau$.

    For every $v \in V(G) \setminus Z$, it holds that $N(v) \subseteq Z$. Note that none of these neighborhoods can be empty since $G$ is connected. For every non-empty subset $Z' \subseteq Z$ let $t(Z')$ be the number of vertices in $V(G) \setminus Z$ whose neighborhood is equal to $Z'$. If $t(Z') > s + 1$, then we remove from $G$ exactly $t(Z') - (s+1)$ vertices of $V(G) \setminus Z$ for each of which the neighborhood is $Z'$. Thus, the resulting graph $G'$ has $\leq s + (s+1) \cdot (2^s - 1) = \O(\tau \cdot 4^\tau)$ vertices.

    First, assume that $G$ has an LDFS ordering $\sigma$ whose $\cl$-tree $T$ has exactly $k$ internal vertices. Due to \cref{lemma:dfs-vc3}, there are at most $s$ internal vertices that are not part of $Z$. In particular, there are at most $s$ internal vertices not in $Z$ for every possible neighborhood in~$Z$. Thus, we can assume that all internal vertices of $T$ are elements of $G'$ as otherwise we may replace vertices by false twins of them. This implies that $T' = T[V(G')]$ is a subtree of~$T$. We claim that $T'$ is the $\cl$-tree of the LDFS ordering $\sigma' = \sigma[V(G')]$ of $G'$. It is straightforward to verify that $T'$ is the \cl-tree of $\sigma'$. This holds since only leaves of $T$ are removed and thus the parent of every vertex is still its rightmost neighbor before it. It remains to show that $\sigma'$ is an LDFS ordering of $G'$. Let $a$, $b$, and $c$ be three vertices with $a \prec_{\sigma'} b \prec_{\sigma'} c$ such that $ac \in E(G)$ and $ab \notin E(G)$. Due to \cref{lemma:4point-ldfs}, there is a $d \in V(G)$ with $a \prec_\sigma d \prec_\sigma b$ and $db \in E(G)$, but $dc \notin E(G)$. As $d$ has a neighbor to the right of it in $\sigma$, vertex $d$ is not a leaf in $T$, due to \cref{obs:dfs-leaves}. Hence, $d$ is also contained in $\sigma'$. This shows that the four point condition also holds for $\sigma'$.

    Now assume that $G'$ has an LDFS ordering $\sigma'$ whose $\cl$-tree $T'$ has exactly $k$ internal vertices. Let $v \in V(G) \setminus V(G')$, then $v \notin Z$ and there are $s+1$ vertices in $G'$ that are false twins of $v$. Due to \cref{lemma:dfs-vc3}, there is at least one of these false twins that is a leaf in $T'$, say $w$. Adding $v$ directly after $w$ creates an LDFS ordering of $G' \cup \{v\}$ whose \cl-tree has the same number of internal vertices due to \cref{lemma:adding-twin}. By repeating this process, we conclude that $G$ has an LDFS $\cl$-tree with exactly $k$ internal vertices.
\end{proof}

\begin{theorem}\label{lemma:ldfs-kernel}
    \MinInL{} and \MaxInL{} of LDFS admit kernels with $\O(k \cdot 4^{k})$ vertices.
\end{theorem}

\begin{proof}
    For \MaxInL{}, we run an LDFS and compute the \cl-tree $T$. If this tree has $\geq k$ internal vertices, the instance is a YES-instance, and we return a $P_{k+1}$; a trivial YES-instance of size $\O(k)$. Otherwise, due to \cref{lemma:dfs-vc1}, the internal vertices of $T$ form a vertex cover, and thus $\tau < k$. So applying \cref{lemma:ldfs-kernel-vc}, we get a kernel with $\O(k \cdot 4^k)$ vertices.

    For \MinInL{}, consider the procedure of \cref{lemma:ldfs-kernel-vc}. If the set $Z$ has size larger then $2k$, then $\tau > k$ and there is no LDFS \cl-tree with $\leq k$ internal vertices, due to \cref{lemma:dfs-vc1}. So we return a $K_{1,k+2}$, a trivial NO-instance of size $\O(k)$. Otherwise, the computed kernel contains $\O(k \cdot 4^k)$ vertices.
\end{proof}

The existence of these kernels implies that there are also \FPT{} algorithms for these problems since we can solve them on the kernel in finite time depending only on $k$. 

\subparagraph{Lower Bound for Kernels} 
As mentioned before, the sizes of our kernels are exponential in the respective parameter, while Bergougnoux et al.~\cite{bergougnoux2025parameterized} present polynomial kernels for the same problems when considered with DFS. Here, we will show that the existence of a polynomial kernel for \MinInL{} of LDFS is unlikely, as it would imply $\NP \subseteq \coNP / \poly$. Our proof uses the following problem.

\newcommand{\CVC}{\textsc{Non-Uniform Clique parameterized by Vertex Cover}}

\begin{problem}[\CVC]~
    \begin{description}
        \item[Instance:] A graph $G$, a vertex cover $Z$ of $G$, an integer $\ell \in \N$
        \item[Parameter:] $|Z|$
        \item[Question:] Is there a clique of $G$ of size $\geq \ell + 1$ that contains a vertex that is not in~$Z$?
    \end{description}
\end{problem}

Bodlaender, Jansen, and Kratsch~\cite{bodlaender2014kernerlization} have shown that \textsc{Non-Uniform Clique parameterized by Vertex Cover} does not admit a polynomial compression, unless $\NP \subseteq \coNP / \poly$. Note that the authors did not explicitly force a vertex outside of $Z$ to be part of the clique. However, in their proof, every possible solution has to fulfill this property.

We call an instance $(G,Z,\ell)$ of \CVC{} \emph{normalized} if it fulfills the following condition: For every two distinct vertices $u,v \in V(G) \setminus Z$, there are vertices $x,y \in Z$ such that $x \in N(u) \setminus N(v)$ and $y \in N(v) \setminus N(u)$. If this is not the case for some pair $u$ and $v$, then the neighborhood of one of the vertices, say $u$, is a subset of the neighborhood of the other. Thus, $u$ is not interesting for the clique problem as for every clique that contains $u$ there is a clique of the same size that contains $v$ instead of $u$. Hence, we can restrict ourselves to normalized instances. Furthermore, we can normalize any instance in polynomial time. Giving these observations, we can prove the following theorem.

\begin{theorem}\label{thm:mininl-lower-bound}
    \MinInL{} of LDFS does not admit a polynomial compression when parameterized by the solution size plus the vertex cover number.
\end{theorem}

\begin{proof}
We show this by giving a polynomial-time \FPT{} reduction from \CVC{} that increases the parameter value only polynomially.\footnote{These reductions are also called polynomial parameter transformations.}  

Let $(G,Z,\ell)$ be a normalized instance of \CVC{}. Let $Z = \{z_1, \dots, z_t\}$ and let $\nZ = V(G) \setminus Z$. We may assume w.l.o.g.~that $G$ is connected and that $1 < \ell < |Z|$. We define a graph $G'$ by adding the following vertices and edges to $G$ (see \cref{fig:lower-bound-1,fig:lower-bound-2,fig:lower-bound-3} for illustrations). First, we add the set of vertices $Y = \{y_1, \dots, y_t\}$, $A = \{a_1, \dots, a_t\}$, $A' = \{a'_1, \dots, a'_t\}$, and $B = \{b_1, \dots, b_t\}$. Additionally, we add the vertices $U = \{u_1, u_2, u_3\}$ as well as the vertices $S = \{s_1, s_2, s_3\}$. We have the following adjacencies.

\begin{figure}
\begin{center}
    \begin{tikzpicture}[scale=0.3]

  \coordinate (AsblobCoord) at (15.0,32);
  \coordinate (BblobCoord) at (3.0,25);
  \coordinate (AblobCoord) at (9.0,25);
  \coordinate (YblobCoord) at (15.0,18.5);
  \coordinate (ZblobCoord) at (6.0,18.5);
  \coordinate (ZqblobCoord) at (-3.0,18.5);
  \coordinate (UblobCoord) at (17.0,7.0);
  \coordinate (SblobCoord) at (0.0,7.0);
  
  \node [rectangle,text width=3.1cm,minimum height=1.0cm, rounded corners,drop shadow] (Asblob) at (AsblobCoord) {};
  \filldraw[lipicsLightGray][] (Asblob.south west)
        [rounded corners=4pt] -- (Asblob.south east)
        [rounded corners=4pt] -- (Asblob.north east)
        [rounded corners=4pt] -- (Asblob.north west)--cycle
        ;
  \node[vertex, label=right:$a_1'$] (a1s) at ($ (AsblobCoord) + (0.3,0) $) {};
  \node at ($ (AsblobCoord) + (-3,0) $) {$A'$};
  
  \node [rectangle,text width=1.1cm,minimum height=0.8cm, rounded corners,drop shadow] (Bblob) at (BblobCoord) {};
  \filldraw[lipicsLightGray][] (Bblob.south west)
        [rounded corners=4pt] -- (Bblob.south east)
        [rounded corners=4pt] -- (Bblob.north east)
        [rounded corners=4pt] -- (Bblob.north west)--cycle
        ;
  \node at (BblobCoord) {$B$};

  \node [rectangle,text width=1.1cm,minimum height=0.8cm, rounded corners,drop shadow] (Ablob) at (AblobCoord) {};
  \filldraw[lipicsLightGray][] (Ablob.south west)
        [rounded corners=4pt] -- (Ablob.south east)
        [rounded corners=4pt] -- (Ablob.north east)
        [rounded corners=4pt] -- (Ablob.north west)--cycle
        ;
  \node at (AblobCoord) {$A$};

  \node [rectangle,text width=1.1cm,minimum height=2.0cm, rounded corners,drop shadow] (Yblob) at (YblobCoord) {};
  \filldraw[lipicsLightGray][] (Yblob.south west)
        [rounded corners=4pt] -- (Yblob.south east)
        [rounded corners=4pt] -- (Yblob.north east)
        [rounded corners=4pt] -- (Yblob.north west)--cycle
        ;
  \node at (YblobCoord) {$Y$};
        
  \node [rectangle,text width=1.1cm,minimum height=0.8cm, rounded corners,drop shadow] (Zblob) at (ZblobCoord) {};
  \filldraw[lipicsLightGray][] (Zblob.south west)
        [rounded corners=4pt] -- (Zblob.south east)
        [rounded corners=4pt] -- (Zblob.north east)
        [rounded corners=4pt] -- (Zblob.north west)--cycle
        ;
   \node at (ZblobCoord) {$Z$};
        
  \node [rectangle,text width=1.1cm,minimum height=2.0cm, rounded corners,drop shadow] (Zqblob) at (ZqblobCoord) {};
  \filldraw[lipicsLightGray][] (Zqblob.south west)
        [rounded corners=4pt] -- (Zqblob.south east)
        [rounded corners=4pt] -- (Zqblob.north east)
        [rounded corners=4pt] -- (Zqblob.north west)--cycle
        ;
   \node at (ZqblobCoord) {$\overline{Z}$};

  \node [rectangle,text width=3.1cm,minimum height=1.0cm, rounded corners,drop shadow] (Ublob) at (UblobCoord) {};
  \filldraw[lipicsLightGray][] (Ublob.south west)
        [rounded corners=4pt] -- (Ublob.south east)
        [rounded corners=4pt] -- (Ublob.north east)
        [rounded corners=4pt] -- (Ublob.north west)--cycle
        ;
  \node[vertex, label=below:$u_1$] (u1) at ($ (UblobCoord) + (-4.5,0) $) {};
  \node[vertex, label=below:$u_2$] (u2) at ($ (UblobCoord) + (-1,0) $) {};
  \node[vertex, label=below:$u_3$] (u3) at ($ (UblobCoord) + (2.5,0) $) {};
  \node at ($(UblobCoord) + (4.5,0) $) {$U$};

  \node [rectangle,text width=3.1cm,minimum height=1.0cm, rounded corners,drop shadow] (Sblob) at (SblobCoord) {};
  \filldraw[lipicsLightGray][] (Sblob.south west)
        [rounded corners=4pt] -- (Sblob.south east)
        [rounded corners=4pt] -- (Sblob.north east)
        [rounded corners=4pt] -- (Sblob.north west)--cycle
        ;
  \node[vertex, label=below:$s_3$] (s3) at ($ (SblobCoord) + (-2.5,0) $) {};
  \node[vertex, label=below:$s_2$] (s2) at ($ (SblobCoord) + (1,0) $) {};
  \node[vertex, label=below:$s_1$] (s1) at ($ (SblobCoord) + (4.5,0) $) {};
  \node at ($(SblobCoord) + (-4.5,0) $) {$S$};

    \draw[blobedge] (Asblob) -- (Bblob);
    \draw[blobedge] (Asblob) -- (Ablob);
    \draw[cliqueedge] (Asblob) -- (Yblob);
    \draw[cliqueedge] (Asblob) to[bend left=40] (u3);
    \draw[cliqueedge] (Ablob) to[bend left=60] (u3);
    \draw[cliqueedge] (Yblob) to (Ublob);
    \draw[blobedge] (Yblob) to (Zblob);
    \draw[cliqueedge] (Yblob) to (Ablob);
    \draw[cliqueedge] (Yblob) to (Bblob);
    \draw[blobedge] (Yblob) to[bend right=25] (Zqblob);
    \draw[cliqueedge] (Ublob) to (Zqblob);
    \draw[cliqueedge] (Ublob) to (Zblob);
    \draw[blobedge] (Bblob) -- (Zblob);
    \draw[blobedge] (Ablob) -- (Zblob);
    \draw[blobedge] (Zqblob) -- (Zblob);
    \draw[cliqueedge] (Zqblob) -- (s3);
    \draw[cliqueedge] (Zqblob) -- (s2);
    \draw[cliqueedge] (Zblob) -- (s2);
    \draw[cliqueedge] (Zblob) -- (s1);
    \draw[cliqueedge] (Zqblob) -- (s1);
    \draw[cliqueedge] (Ublob) -- (s1);

    \draw[normaledge] (a1s) to[bend left=20] (u2);
    \draw[normaledge] (u1) to (u2);
    \draw[normaledge] (u2) to (u3);
    \draw[normaledge] (u1) to[bend left=30] (u3);
    \draw[normaledge] (s1) to (s2);
    \draw[normaledge] (s2) to (s3);
    \draw[normaledge] (s1) to[bend right=30] (s3);

\end{tikzpicture}

\end{center}

    \caption{A schematic sketch of the graph $G'$ used in the proof of \cref{thm:mininl-lower-bound}. Yellow thick edges describe that there is an edge between every pair of vertices of the two sets. Gray thin edges describe that there are some edges between vertices of the two sets. For the details of these connections see \cref{fig:lower-bound-2,fig:lower-bound-3}.}
    \label{fig:lower-bound-1}
\end{figure}
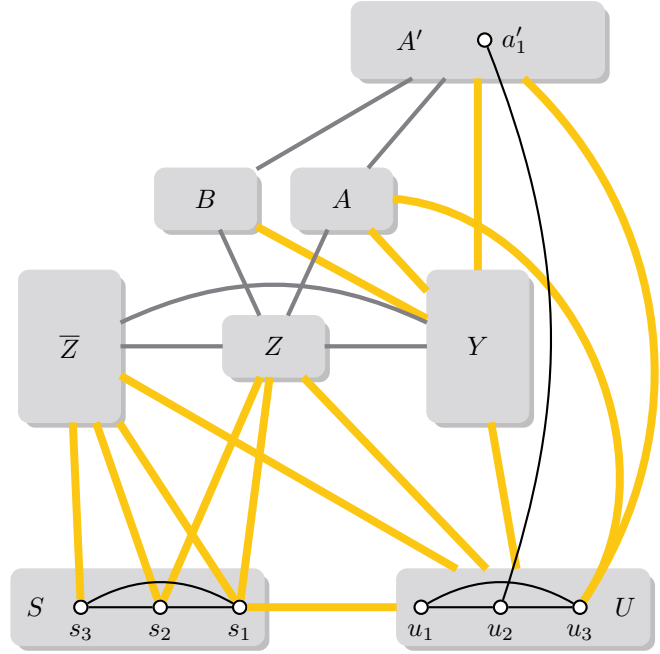

\begin{enumerate}
    \item $y_i z_j \in E(G')$ if and only if $i \neq j$,
    \item for all $v \in \nZ$, $y_i v \in E(G')$ if and only if $z_i \notin N_G(v)$,
    \item $a_i z_j, b_i z_j, a_i a'_j \in E(G')$ if and only if $i=j$,
    \item $b_i a'_j \in E(G')$ if and only if $j \in \{i,i+1\}$,
    \item $G'[Y]$, $G'[S]$, and $G'[U]$ are cliques,
    \item $Y \cup Z \cup \nZ \subseteq N(u_1) \cap N(u_2) \cap N(u_3)$,
    \item $A \cup A' \subseteq N(u_3)$,
    \item $Y \subseteq N(a_i) \cap N(a'_i) \cap N(b_i)$ for all $i \in [t]$
    \item $a'_1u_2 \in E(G')$,
    \item $N_{G'}(s_1) = Z \cup \nZ \cup U$, $N_{G'}(s_2) = Z \cup \nZ$, $N_{G'}(s_3) = \nZ$.
\end{enumerate}

We furthermore append a leaf to every vertex in $Z \cup Y \cup A' \cup B \cup U \cup S$. Therefore, these vertices are internal vertices in every $\cl$-tree, which implies the following.

\begin{claim}\label{obs:lower-bound}
    The $\cl$-tree of every LDFS ordering of $G'$ has at least $4t + 6$ internal vertices.
\end{claim}

We now show that the number of internal vertices of the \cl-tree of LDFS orderings of $G'$ is related to the size of a clique containing a vertex of $\nZ$ as long as the LDFS ordering starts with the vertices of $S$.

\begin{claim}\label{lemma:kernel-lower-bound}
    Every LDFS ordering $\sigma$ of $G'$ that starts with the vertices of $S$ visits afterwards a maximal clique of $G$ containing a vertex of $\nZ$. Let $\lambda + 1$ be the size of that clique. Then the \cl-tree of $\sigma$ has $5t - \lambda + 7$ internal vertices.
    Furthermore, for every maximal clique $C$ of $G$ containing a vertex of $\nZ$ there is an LDFS ordering that starts with the vertices of $S$ and then visits the vertices of $C$. 
\end{claim}

\begin{figure}[t]
\begin{center}
\resizebox{0.75\textwidth}{!}{\begin{tikzpicture}[scale=0.5]

  \coordinate (YblockCoord) at (32,3.0);
  \coordinate (ZpblockCoord) at (32,8.0);
  \coordinate (ZblockCoord) at (32,13.0);

  \node [rectangle,text width=10.0cm,minimum height=1.1cm, rounded corners,drop shadow] (Yblock) at (YblockCoord) {};
  \filldraw[lipicsLightGray][] (Yblock.south west)
        [rounded corners=4pt] -- (Yblock.south east)
        [rounded corners=4pt] -- (Yblock.north east)
        [rounded corners=4pt] -- (Yblock.north west)--cycle
        ;
  \node[vertex,opacity=.2] (y1) at ($ (YblockCoord) + (-8,0) $) {};
  \node[vertex,opacity=.2] (y2) at ($ (YblockCoord) + (-7,0) $) {};
  \node[vertex,opacity=.2] (y3) at ($ (YblockCoord) + (-6,0) $) {};
  \node[vertex,opacity=.2] (y4) at ($ (YblockCoord) + (-5,0) $) {};
  \node[vertex,opacity=.2] (y5) at ($ (YblockCoord) + (-4,0) $) {};
  \node[vertex,opacity=.2] (y6) at ($ (YblockCoord) + (-3,0) $) {};
  \node[vertex,opacity=.2] (y7) at ($ (YblockCoord) + (-2,0) $) {};
  \node[vertex,opacity=.2] (y8) at ($ (YblockCoord) + (-1,0) $) {};
  \node[vertex, label=below:$y_i$] (yi) at ($ (YblockCoord) + (0.0,0) $) {};
  \node[vertex,opacity=.2] (y9) at ($ (YblockCoord) + (1,0) $) {};
  \node[vertex,opacity=.2] (y10) at ($ (YblockCoord) + (2,0) $) {};
  \node[vertex,opacity=.2] (y11) at ($ (YblockCoord) + (3,0) $) {};
  \node[vertex,opacity=.2] (y12) at ($ (YblockCoord) + (4,0) $) {};
  \node[vertex,opacity=.2] (y13) at ($ (YblockCoord) + (5,0) $) {};
  \node[vertex,opacity=.2] (y14) at ($ (YblockCoord) + (6,0) $) {};
  \node[vertex,opacity=.2] (y15) at ($ (YblockCoord) + (7,0) $) {};
  \node[vertex,opacity=.2] (y16) at ($ (YblockCoord) + (8,0) $) {};
  \node at ($ (YblockCoord) + (9.5,0) $) {$Y$};

  \node [rectangle,text width=10.0cm,minimum height=1.1cm, rounded corners,drop shadow] (Zpblock) at (ZpblockCoord) {};
  \filldraw[lipicsLightGray][] (Zpblock.south west)
        [rounded corners=4pt] -- (Zpblock.south east)
        [rounded corners=4pt] -- (Zpblock.north east)
        [rounded corners=4pt] -- (Zpblock.north west)--cycle
        ;
  \node[vertex] (b1) at ($ (ZpblockCoord) + (8,0.0) $) {};
  \node[vertex] (b2) at ($ (ZpblockCoord) + (6,0.0) $) {};
  \node[vertex] (b3) at ($ (ZpblockCoord) + (4,0.0) $) {};
  \node[vertex] (b4) at ($ (ZpblockCoord) + (2,0.0) $) {};
  \node[vertex] (b5) at ($ (ZpblockCoord) + (0,0.0) $) {};
  \node[vertex] (b6) at ($ (ZpblockCoord) + (-2,0.0) $) {};
  \node[vertex] (b7) at ($ (ZpblockCoord) + (-4,0.0) $) {};
  \node[vertex] (b8) at ($ (ZpblockCoord) + (-6,0.0) $) {};
  \node[vertex] (b9) at ($ (ZpblockCoord) + (-8,0.0) $) {};
  \node at ($ (ZpblockCoord) + (9.5,0) $) {$\overline{Z}$};

  \node [rectangle,text width=10.0cm,minimum height=1.1cm, rounded corners,drop shadow] (Zblock) at (ZblockCoord) {};
  \filldraw[lipicsLightGray][] (Zblock.south west)
        [rounded corners=4pt] -- (Zblock.south east)
        [rounded corners=4pt] -- (Zblock.north east)
        [rounded corners=4pt] -- (Zblock.north west)--cycle
        ;
  \node[vertex,opacity=.2] (z1) at ($ (ZblockCoord) + (-8,0) $) {};
  \node[vertex,opacity=.2] (z2) at ($ (ZblockCoord) + (-7,0) $) {};
  \node[vertex,opacity=.2] (z3) at ($ (ZblockCoord) + (-6,0) $) {};
  \node[vertex,opacity=.2] (z4) at ($ (ZblockCoord) + (-5,0) $) {};
  \node[vertex,opacity=.2] (z5) at ($ (ZblockCoord) + (-4,0) $) {};
  \node[vertex,opacity=.2] (z6) at ($ (ZblockCoord) + (-3,0) $) {};
  \node[vertex,opacity=.2] (z7) at ($ (ZblockCoord) + (-2,0) $) {};
  \node[vertex,opacity=.2] (z8) at ($ (ZblockCoord) + (-1,0) $) {};
  \node[vertex, label=above:$z_i$] (zi) at ($ (ZblockCoord) + (0,0.0) $) {};
  \node[vertex,opacity=.2] (z9) at ($ (ZblockCoord) + (1,0) $) {};
  \node[vertex,opacity=.2] (z10) at ($ (ZblockCoord) + (2,0) $) {};
  \node[vertex,opacity=.2] (z11) at ($ (ZblockCoord) + (3,0) $) {};
  \node[vertex,opacity=.2] (z12) at ($ (ZblockCoord) + (4,0) $) {};
  \node[vertex,opacity=.2] (z13) at ($ (ZblockCoord) + (5,0) $) {};
  \node[vertex,opacity=.2] (z14) at ($ (ZblockCoord) + (6,0) $) {};
  \node[vertex,opacity=.2] (z15) at ($ (ZblockCoord) + (7,0) $) {};
  \node[vertex,opacity=.2] (z16) at ($ (ZblockCoord) + (8,0) $) {};
  \node at ($ (ZblockCoord) + (9.5,0) $) {$Z$};

    \draw[thinedge] (b1) -- (y1);
    \draw[thinedge] (b2) -- (z1);
    \draw[thinedge] (b3) -- (z1);
    \draw[thinedge] (b4) -- (y1);
    \draw[thinedge] (b5) -- (y1);
    \draw[thinedge] (b6) -- (z1);
    \draw[thinedge] (b7) -- (y1);
    \draw[thinedge] (b8) -- (z1);
    \draw[thinedge] (b9) -- (z1);

    \draw[thinedge] (b1) -- (y2);
    \draw[thinedge] (b2) -- (y2);
    \draw[thinedge] (b3) -- (z2);
    \draw[thinedge] (b4) -- (y2);
    \draw[thinedge] (b5) -- (z2);
    \draw[thinedge] (b6) -- (y2);
    \draw[thinedge] (b7) -- (z2);
    \draw[thinedge] (b8) -- (y2);
    \draw[thinedge] (b9) -- (z2);

    \draw[thinedge] (b1) -- (y3);
    \draw[thinedge] (b2) -- (y3);
    \draw[thinedge] (b3) -- (z3);
    \draw[thinedge] (b4) -- (y3);
    \draw[thinedge] (b5) -- (z3);
    \draw[thinedge] (b6) -- (y3);
    \draw[thinedge] (b7) -- (z3);
    \draw[thinedge] (b8) -- (y3);
    \draw[thinedge] (b9) -- (z3);

    \draw[thinedge] (b1) -- (y4);
    \draw[thinedge] (b2) -- (y4);
    \draw[thinedge] (b3) -- (z4);
    \draw[thinedge] (b4) -- (y4);
    \draw[thinedge] (b5) -- (z4);
    \draw[thinedge] (b6) -- (y4);
    \draw[thinedge] (b7) -- (z4);
    \draw[thinedge] (b8) -- (y4);
    \draw[thinedge] (b9) -- (z4);

    \draw[thinedge] (b1) -- (y5);
    \draw[thinedge] (b2) -- (z5);
    \draw[thinedge] (b3) -- (y5);
    \draw[thinedge] (b4) -- (y5);
    \draw[thinedge] (b5) -- (y5);
    \draw[thinedge] (b6) -- (z5);
    \draw[thinedge] (b7) -- (y5);
    \draw[thinedge] (b8) -- (y5);
    \draw[thinedge] (b9) -- (z5);

    \draw[thinedge] (b1) -- (y6);
    \draw[thinedge] (b2) -- (z6);
    \draw[thinedge] (b3) -- (y6);
    \draw[thinedge] (b4) -- (y6);
    \draw[thinedge] (b5) -- (y6);
    \draw[thinedge] (b6) -- (z6);
    \draw[thinedge] (b7) -- (y6);
    \draw[thinedge] (b8) -- (y6);
    \draw[thinedge] (b9) -- (z6);

    \draw[thinedge] (b1) -- (z7);
    \draw[thinedge] (b2) -- (y7);
    \draw[thinedge] (b3) -- (y7);
    \draw[thinedge] (b4) -- (y7);
    \draw[thinedge] (b5) -- (z7);
    \draw[thinedge] (b6) -- (y7);
    \draw[thinedge] (b7) -- (z7);
    \draw[thinedge] (b8) -- (z7);
    \draw[thinedge] (b9) -- (y7);

    \draw[thinedge] (b1) -- (y8);
    \draw[thinedge] (b2) -- (y8);
    \draw[thinedge] (b3) -- (z8);
    \draw[thinedge] (b4) -- (y8);
    \draw[thinedge] (b5) -- (y8);
    \draw[thinedge] (b6) -- (y8);
    \draw[thinedge] (b7) -- (z8);
    \draw[thinedge] (b8) -- (y8);
    \draw[thinedge] (b9) -- (z8);

    \draw[thinedge] (b1) -- (z9);
    \draw[thinedge] (b2) -- (y9);
    \draw[thinedge] (b3) -- (y9);
    \draw[thinedge] (b4) -- (z9);
    \draw[thinedge] (b5) -- (y9);
    \draw[thinedge] (b6) -- (y9);
    \draw[thinedge] (b7) -- (y9);
    \draw[thinedge] (b8) -- (z9);
    \draw[thinedge] (b9) -- (z9);

    \draw[thinedge] (b1) -- (y10);
    \draw[thinedge] (b2) -- (y10);
    \draw[thinedge] (b3) -- (z10);
    \draw[thinedge] (b4) -- (y10);
    \draw[thinedge] (b5) -- (y10);
    \draw[thinedge] (b6) -- (z10);
    \draw[thinedge] (b7) -- (z10);
    \draw[thinedge] (b8) -- (z10);
    \draw[thinedge] (b9) -- (z10);

    \draw[thinedge] (b1) -- (z11);
    \draw[thinedge] (b2) -- (z11);
    \draw[thinedge] (b3) -- (z11);
    \draw[thinedge] (b4) -- (y11);
    \draw[thinedge] (b5) -- (z11);
    \draw[thinedge] (b6) -- (y11);
    \draw[thinedge] (b7) -- (z11);
    \draw[thinedge] (b8) -- (y11);
    \draw[thinedge] (b9) -- (y11);

    \draw[thinedge] (b1) -- (y12);
    \draw[thinedge] (b2) -- (y12);
    \draw[thinedge] (b3) -- (z12);
    \draw[thinedge] (b4) -- (y12);
    \draw[thinedge] (b5) -- (y12);
    \draw[thinedge] (b6) -- (z12);
    \draw[thinedge] (b7) -- (y12);
    \draw[thinedge] (b8) -- (z12);
    \draw[thinedge] (b9) -- (z12);

    \draw[thinedge] (b1) -- (z13);
    \draw[thinedge] (b2) -- (z13);
    \draw[thinedge] (b3) -- (z13);
    \draw[thinedge] (b4) -- (y13);
    \draw[thinedge] (b5) -- (y13);
    \draw[thinedge] (b6) -- (y13);
    \draw[thinedge] (b7) -- (z13);
    \draw[thinedge] (b8) -- (z13);
    \draw[thinedge] (b9) -- (z13);

    \draw[thinedge] (b1) -- (y14);
    \draw[thinedge] (b2) -- (z14);
    \draw[thinedge] (b3) -- (y14);
    \draw[thinedge] (b4) -- (z14);
    \draw[thinedge] (b5) -- (y14);
    \draw[thinedge] (b6) -- (y14);
    \draw[thinedge] (b7) -- (y14);
    \draw[thinedge] (b8) -- (z14);
    \draw[thinedge] (b9) -- (z14);

    \draw[thinedge] (b1) -- (y15);
    \draw[thinedge] (b2) -- (y15);
    \draw[thinedge] (b3) -- (z15);
    \draw[thinedge] (b4) -- (z15);
    \draw[thinedge] (b5) -- (z15);
    \draw[thinedge] (b6) -- (y15);
    \draw[thinedge] (b7) -- (z15);
    \draw[thinedge] (b8) -- (y15);
    \draw[thinedge] (b9) -- (z15);

    \draw[thinedge] (b1) -- (z16);
    \draw[thinedge] (b2) -- (y16);
    \draw[thinedge] (b3) -- (z16);
    \draw[thinedge] (b4) -- (y16);
    \draw[thinedge] (b5) -- (z16);
    \draw[thinedge] (b6) -- (y16);
    \draw[thinedge] (b7) -- (y16);
    \draw[thinedge] (b8) -- (z16);
    \draw[thinedge] (b9) -- (z16);

    \draw[exampleedge] (b1) -- (zi);
    \draw[exampleedge] (b2) -- (zi);
    \draw[exampleedge] (b3) -- (yi);
    \draw[exampleedge] (b4) -- (zi);
    \draw[exampleedge] (b5) -- (yi);
    \draw[exampleedge] (b6) -- (yi);
    \draw[exampleedge] (b7) -- (zi);
    \draw[exampleedge] (b8) -- (yi);
    \draw[exampleedge] (b9) -- (zi);

\end{tikzpicture}

}
\end{center}
    \caption{Sketch of the connections between $\nZ$ and the sets $Z$ and~$Y$ in the graph $G'$ used in the proof of \cref{thm:mininl-lower-bound}.}\label{fig:lower-bound-2}
\end{figure}
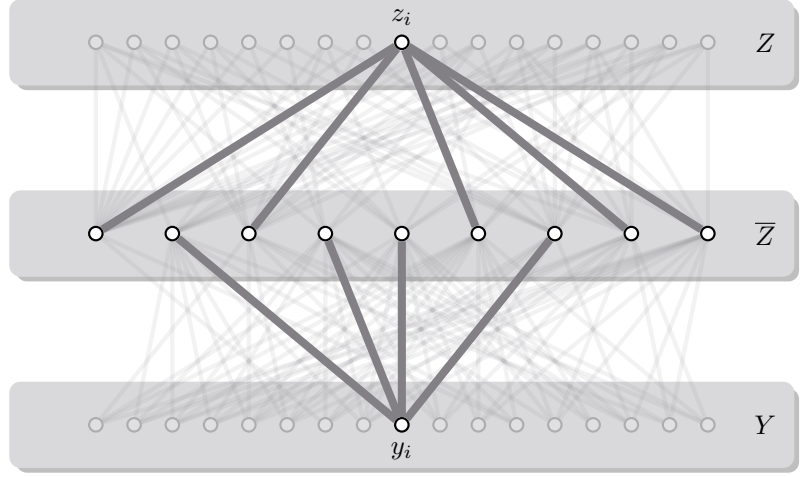
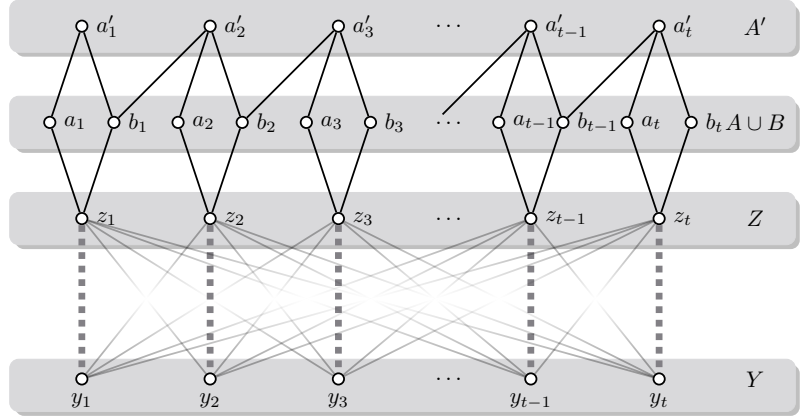
\begin{figure}[t]
\begin{center}
    \resizebox{0.75\textwidth}{!}{\begin{tikzpicture}[scale=0.5, xshift=-1cm]

  \tikzfading[name=fade out, inner color=transparent!0, outer color=transparent!100]

  \tikzfading[name=fade ends,
  inner color=transparent!100,
  middle color=transparent!100,
  outer color=transparent!70]

  \coordinate (L1Coord) at (3.0,11);
  \coordinate (L2Coord) at (2.0,8);
  \coordinate (L3Coord) at (3.0,5);
  \coordinate (L4Coord) at (3.0,0);

  \node [rectangle,text width=12.0cm,minimum height=0.8cm, rounded
  corners,drop shadow] (Yblock) at ($ (L1Coord) + (10,0) $) {};
  \filldraw[lipicsLightGray][] (Yblock.south west)
        [rounded corners=4pt] -- (Yblock.south east)
        [rounded corners=4pt] -- (Yblock.north east)
        [rounded corners=4pt] -- (Yblock.north west)--cycle
        ;
  \node at ($ (L1Coord) + (21,0) $) {$A'$};
  \node[vertex, label=right:$a_1'$] (a1s) at ($ (L1Coord) + (0.0,0.0) $) {};
  \node[vertex, label=right:$a_2'$] (a2s) at ($ (L1Coord) + (4.0,0.0) $) {};
  \node[vertex, label=right:$a_3'$] (a3s) at ($ (L1Coord) + (8.0,0.0) $) {};
  \node (pp) at ($ (L1Coord) + (11.5,0.0) $) {\dots};
  \node[vertex, label=right:$a_{t-1}'$] (atm1s) at ($ (L1Coord) + (14,0.0) $) {};
  \node[vertex, label=right:$a_t'$] (ats) at ($ (L1Coord) + (18,0.0) $) {};

  \node [rectangle,text width=12.0cm,minimum height=0.8cm, rounded
  corners,drop shadow] (Yblock) at ($ (L2Coord) + (11,0) $) {};
  \filldraw[lipicsLightGray][] (Yblock.south west)
        [rounded corners=4pt] -- (Yblock.south east)
        [rounded corners=4pt] -- (Yblock.north east)
        [rounded corners=4pt] -- (Yblock.north west)--cycle
        ;
  \node at ($ (L2Coord) + (22,0) $) {$A \cup B$};
  \node[vertex, label=right:$a_1$] (a1) at ($ (L2Coord) + (0.0,0.0) $) {};
  \node[vertex, label=right:$b_1$] (b1) at ($ (L2Coord) + (2.0,0.0) $) {};
  \node[vertex, label=right:$a_2$] (a2) at ($ (L2Coord) + (4.0,0.0) $) {};
  \node[vertex, label=right:$b_2$] (b2) at ($ (L2Coord) + (6.0,0.0) $) {};
  \node[vertex, label=right:$a_3$] (a3) at ($ (L2Coord) + (8.0,0.0) $) {};
  \node[vertex, label=right:$b_3$] (b3) at ($ (L2Coord) + (10.0,0.0) $) {};
  \node (pp2) at ($ (L2Coord) + (12.5,0.0) $) {\dots};
  \node (kkk) at ($ (L2Coord) + (12,0.0) $) {};
  \node[vertex, label=right:$a_{t-1}$] (atm1) at ($ (L2Coord) + (14,0.0) $) {};
  \node[vertex, label=right:$b_{t-1}$] (btm1) at ($ (L2Coord) + (16,0.0) $) {};
  \node[vertex, label=right:$a_t$] (at) at ($ (L2Coord) + (18,0.0) $) {};
  \node[vertex, label=right:$b_t$] (bt) at ($ (L2Coord) + (20,0.0) $) {};
  
  \node [rectangle,text width=12.0cm,minimum height=0.8cm, rounded
  corners,drop shadow] (Yblock) at ($ (L3Coord) + (10,0) $) {};
  \filldraw[lipicsLightGray][] (Yblock.south west)
        [rounded corners=4pt] -- (Yblock.south east)
        [rounded corners=4pt] -- (Yblock.north east)
        [rounded corners=4pt] -- (Yblock.north west)--cycle
        ;
  \node at ($ (L3Coord) + (21,0) $) {$Z$};
  \node[vertex, label=right:$z_1$] (z1) at ($ (L3Coord) + (0.0,0.0) $) {};
  \node[vertex, label=right:$z_2$] (z2) at ($ (L3Coord) + (4.0,0.0) $) {};
  \node[vertex, label=right:$z_3$] (z3) at ($ (L3Coord) + (8.0,0.0) $) {};
  \node (pp) at ($ (L3Coord) + (11.5,0.0) $) {\dots};
  \node[vertex, label=right:$z_{t-1}$] (ztm1) at ($ (L3Coord) + (14,0.0) $) {};
  \node[vertex, label=right:$z_t$] (zt) at ($ (L3Coord) + (18,0.0) $) {};

  \node [rectangle,text width=12.0cm,minimum height=0.8cm, rounded
  corners,drop shadow] (Yblock) at ($ (L4Coord) + (10,-0.2) $) {};
  \filldraw[lipicsLightGray][] (Yblock.south west)
        [rounded corners=4pt] -- (Yblock.south east)
        [rounded corners=4pt] -- (Yblock.north east)
        [rounded corners=4pt] -- (Yblock.north west)--cycle
        ;
  \node at ($ (L4Coord) + (21,0) $) {$Y$};
  \node[vertex, label=below:$y_1$] (y1) at ($ (L4Coord) + (0.0,0.0) $) {};
  \node[vertex, label=below:$y_2$] (y2) at ($ (L4Coord) + (4.0,0.0) $) {};
  \node[vertex, label=below:$y_3$] (y3) at ($ (L4Coord) + (8.0,0.0) $) {};
  \node (pp) at ($ (L4Coord) + (11.5,0.0) $) {\dots};
  \node[vertex, label=below:$y_{t-1}$] (ytm1) at ($ (L4Coord) + (14,0.0) $) {};
  \node[vertex, label=below:$y_t$] (yt) at ($ (L4Coord) + (18,0.0) $) {};

  \draw[normaledge] (a1s) -- (a1);
  \draw[normaledge] (a1s) -- (b1);
  \draw[normaledge] (a2s) -- (b1);
  \draw[normaledge] (a2s) -- (a2);
  \draw[normaledge] (a2s) -- (b2);
  \draw[normaledge] (a3s) -- (b2);
  \draw[normaledge] (a3s) -- (a3);
  \draw[normaledge] (a3s) -- (b3);
  \draw[normaledge] (atm1s) -- (kkk);
  \draw[normaledge] (atm1s) -- (atm1);
  \draw[normaledge] (atm1s) -- (btm1);
  \draw[normaledge] (ats) -- (btm1);
  \draw[normaledge] (ats) -- (at);
  \draw[normaledge] (ats) -- (bt);

  \draw[normaledge] (a1) -- (z1);
  \draw[normaledge] (b1) -- (z1);
  \draw[normaledge] (a2) -- (z2);
  \draw[normaledge] (b2) -- (z2);
  \draw[normaledge] (a3) -- (z3);
  \draw[normaledge] (b3) -- (z3);
  \draw[normaledge] (atm1) -- (ztm1);
  \draw[normaledge] (btm1) -- (ztm1);
  \draw[normaledge] (at) -- (zt);
  \draw[normaledge] (bt) -- (zt);

  \draw[noneedge] (z1) -- (y1);
  \draw[noneedge] (z2) -- (y2);
  \draw[noneedge] (z3) -- (y3);
  \draw[noneedge] (ztm1) -- (ytm1);
  \draw[noneedge] (zt) -- (yt);

  \draw[fadeedge] (z1) -- (y2);
  \draw[fadeedge] (z1) -- (y3);
  \draw[fadeedge] (z1) -- (ytm1);
  \draw[fadeedge] (z1) -- (yt);

  \draw[fadeedge] (z2) -- (y1);
  \draw[fadeedge] (z2) -- (y3);
  \draw[fadeedge] (z2) -- (ytm1);
  \draw[fadeedge] (z2) -- (yt);

  \draw[fadeedge] (z3) -- (y1);
  \draw[fadeedge] (z3) -- (y2);
  \draw[fadeedge] (z3) -- (ytm1);
  \draw[fadeedge] (z3) -- (yt);

  \draw[fadeedge] (ztm1) -- (y1);
  \draw[fadeedge] (ztm1) -- (y2);
  \draw[fadeedge] (ztm1) -- (y3);
  \draw[fadeedge] (ztm1) -- (yt);

  \draw[fadeedge] (zt) -- (y1);
  \draw[fadeedge] (zt) -- (y2);
  \draw[fadeedge] (zt) -- (y3);
  \draw[fadeedge] (zt) -- (ytm1);

\end{tikzpicture}

}
\end{center}
 \caption{Connections between the sets $Z$, $Y$, $A$, $A'$, and $B$ in the graph $G'$ used in the proof of \cref{thm:mininl-lower-bound}.}\label{fig:lower-bound-3}
\end{figure}

\begin{claimproof}
    First note that in the following argumentation, we will ignore the leaves that have been appended to certain vertices, i.e., we do not mention when they are visited by the LDFS as this is not relevant for the number of internal vertices. Further note that the only vertices of $V(G') \setminus \nZ$ that are not forced to be internal vertices via appended leaves are the vertices of $A$. Thus, these vertices will be the part of the graph where we ``count'' the size of the clique. The overall idea is that those vertices that are part of the clique $C$ are visited to early and, therefore, cannot be children of their respective vertices in $A$.
    
    Consider an LDFS ordering $\sigma$ of $G'$ starting with the vertices of $S$. The only vertices of $G'$ that are adjacent to all vertices of $S$ are the vertices in $\nZ$. Hence, $\sigma$ visits a vertex $v$ of $\nZ$ directly after $S$. Note that $v$ could be any vertex of $\nZ$ as all of them have the same set of visited neighbors. No unvisited vertex that is adjacent to $s_3$ is also adjacent to $v$. Hence, we will visit elements of $N_{G'}(v) \cap N_{G'}(s_1) \cap N_{G'}(s_2)$ as long as these vertices induce a clique. All vertices fulfilling this condition are contained in $Z$. Therefore, $\sigma$ visits a maximal clique $C$ of $G'$ containing $v$ and any of these cliques would be a possible option. W.l.o.g.~we may assume that $C = \{v, z_1, \dots, z_\lambda\}$ and the vertices of $C$ are visited in the given order. Afterwards, every remaining vertex of $Z \cup \nZ$ is non-adjacent to at least one of the vertices visited after $S$ since $C$ is maximal. 

    We claim that $\sigma$ will now visit the vertices of $U$ and then the vertices of $Y$. To see this, we define the following two sets: $Y_1 := \{y_i \mid z_i \notin N(v)\}$ and $Y_2 := \{y_i \mid z_i \in N(v), z_i \notin C\}$. Note that $Y \setminus (Y_1 \cup Y_2) = \{y_1, \dots, y_\lambda\}$. We observe that all vertices in $Y_1 \cup U$ are adjacent to all vertices that are visited after $S$. All other unvisited vertices have at least one non-neighbor that was visited after $S$: the unvisited vertices of $Z \cup \nZ$ are non-adjacent to at least one vertex of $C$ and the vertices in $(Y \setminus Y_1) \cup A \cup A' \cup B$ are non-adjacent to $v$.
    Furthermore, the vertices in $U$ have a neighbor in $S$ while the vertices of $Y_1$ do not. Hence, $\sigma$ first visits $U$ in an arbitrary order and, afterwards, it visits $Y_1$ in an arbitrary order. 
    
    Now, the vertices of $Y_2$ are adjacent to all vertices visited after $v$. All other unvisited vertices have at least one non-neighbor that was visited after $v$:
    
    \begin{itemize}
        \item If an unvisited vertex of $Z$ is adjacent to $v$, then it is not adjacent to at least one vertex of $Z \cap C$ since $C$ is maximal.
        \item If an unvisited vertex $z_i$ of $Z$ is not adjacent to $v$, then $y_i \in Y_1$ and $y_iz_i \notin E(G)$.
        \item Every vertex of $\nZ \setminus \{v\}$ has a non-neighbor in $Y_1$ since we have considered a normalized instance of \CVC{} and, thus, every vertex of $\nZ \setminus \{v\}$ has some neighbor in $Z$ that is not a neighbor of $v$. 
        \item The vertices in $Y \setminus (Y_1 \cup Y_2) = \{y_1, \dots, y_\lambda\}$ are not adjacent to one vertex of $Z \cap C$.
        \item The vertices in $A \cup A' \cup B$ are not adjacent to $u_1$.
    \end{itemize}
    
    Thus, $\sigma$ visits the vertices of $Y_2$ in an arbitrary order. We claim that $\sigma$ visits the vertices $(y_1, \dots, y_\lambda)$ next in that order. Note that $y_i$ is adjacent to all vertices visited after $z_i$ while all other unvisited vertices are non-adjacent to some vertex visited after $z_i$:
    
    \begin{itemize}
        \item The unvisited vertices of $Z \cup \nZ$ are non-adjacent to at least one vertex of $Y_1 \cup Y_2$.
        \item The vertices in $A \cup A' \cup B$ are not adjacent to $u_1$.
        \item Vertex $y_j$ with $j > i$ is not adjacent to $z_j$.
    \end{itemize}
    
    The only remaining vertices that are adjacent to all vertices of $Y$ are the vertices in $A \cup A' \cup B$. Among these vertices, $a'_1$ is the only vertex that is adjacent both to $u_2$ and $u_3$ while all other vertices are only adjacent to $u_3$ or to no element of $U$. So, $\sigma$ visits $a'_1$ next. 
    
    This vertex has two unvisited neighbors: $a_1$ and $b_1$. The visited neighbors of $b_1$ are also neighbors of $a_1$ and $u_3$ is adjacent to $a_1$ but not to $b_1$. Thus, $\sigma$ visits $a_1$ next. If $z_1$ was already visited, i.e., $z_1 \in C$, then $a_1$ has no unvisited neighbor and, thus, it is a leaf of the $\cl$-tree. In this case, LDFS backtracks to $a'_1$ and then $\sigma$  visits $b_1$. If $z_1$ was not visited, i.e., $z_1 \notin C$, then $\sigma$ visits $z_1$ after $a_1$. Now, $\sigma$ visits $b_1$ since it is the only neighbor of $z_1$ that is adjacent to another vertex of the last three visited vertices. In this case, $a_1$ is not a leaf. In both cases, $\sigma$ visits $a'_2$ next since it is the only unvisited neighbor of $b_1$. Repeating this procedure, we construct $\lambda$ vertices of $A$ that are leaves in the $\cl$-tree. After reaching $b_t$, we have visited all vertices of $G'$ except for the remaining vertices of $\nZ$. Since these vertices form an independent set, they all become leaves of the \cl-tree.

    Finally, let us count the number of internal vertices of the \cl-tree. All vertices of $Z \cup S \cup U \cup Y \cup A' \cup B \cup \{v\}$ are internal vertices. Furthermore, $|A| - \lambda$ vertices of $A$ are internal vertices. Thus, there are $5t - \lambda + 7$ internal vertices.
\end{claimproof}

    To overcome the problem with fixing the start vertices of the LDFS ordering, we consider $t+2$ copies $H_1, \dots, H_{t+2}$ of $G'$. For each $i \in [t+2]$, let $S_i$ be the $S$-sets of copy $H_i$. Let $s$ be another vertex that is adjacent to all vertices in all $S_i$-sets. We call the resulting graph~$H$. We claim that $H$ has an LDFS $\cl$-tree with at most $k = (t+2) \cdot (5t - \ell + 7) + 1$ internal vertices if and only if $G$ has a clique of size $\ell+1$ that contains exactly one vertex of~$\nZ$.

    First assume that there is such a clique. Then, due to \cref{lemma:kernel-lower-bound}, every $H_i$ has some LDFS ordering $\sigma_i$ starting with the vertices of $S_i$ whose $\cl$-tree has $5t - \ell + 7$ internal vertices. We observe that $(s) \oplus \sigma_1 \oplus \dots \oplus \sigma_{t+2}$ forms an LDFS ordering of $H$ whose $\cl$-tree has $(t+2) \cdot (5t - \ell + 7) + 1$ internal vertices (here $\oplus$ stands for the concatenation operator). 

    Now assume that such a clique does not exist, i.e., every clique of $G$ containing a vertex of $\nZ$ contains at most $\ell - 1$ vertices of $Z$. Let $\sigma$ be an LDFS ordering of $H$ and let $\sigma_1, \dots, \sigma_{t+2}$ be the suborderings of $\sigma$ induced by the vertices of $H_1, \dots, H_{t+2}$, respectively. It is easy to see that for every $i \in [t+2]$, the subordering $\sigma_i$ forms an LDFS ordering of $H_i$. Furthermore, all but maybe $\sigma_1$ start with the vertices of set $S_i$. Hence, \cref{lemma:kernel-lower-bound} implies that the \cl-tree of these orderings has at least $5t - \ell + 8$ internal vertices. The $\cl$-tree of $\sigma_1$ has at least $4t + 6$ internal vertices, due to \cref{obs:lower-bound}. Furthermore, $s$ is an internal vertex. Thus, the \cl-tree of $\sigma$ has at least $(t+1) \cdot (5t - \ell + 8) + (4t + 6) + 1 = (t+1) \cdot (5t - \ell + 7) + 5t + 8 > (t+2) \cdot (5t - \ell + 7) + 1$ internal vertices. 

    Note that both the vertex cover number of $G'$ and the value $k$ are bounded by $\O(|Z|^2)$. Hence, the given reduction is a polynomial-time \FPT{} reduction from \CVC{} to \MinInL{} of LDFS parameterized by vertex cover number plus solution size that increases the parameter only polynomially. 
\end{proof}

We can adapt this proof to give the same result also for the maximization problem of internal vertices. In the proof of \cref{thm:mininl-lower-bound}, visiting some vertex of the clique implied that this vertex is missing to make some other vertex to an internal vertex. In the proof for the maximization problem, however, visiting a clique vertex implies that another vertex is visited earlier, ensuring that this other vertex becomes an internal vertex. 

\begin{theorem}\label{thm:maxinl-lower-bound}
    \MaxInL{} of LDFS does not admit a polynomial compression when parameterized by the solution size plus the vertex cover number.
\end{theorem}

\begin{proof}
We adapt the proof of \cref{thm:mininl-lower-bound}. Again, we give a polynomial-time \FPT{} reduction from \CVC{} that increases the parameter value only polynomially.

Let $(G,Z,\ell)$ be a normalized instance of \CVC{}. Let $Z = \{z_1, \dots, z_t\}$ and let $\nZ = V(G) \setminus Z$. We may assume w.l.o.g.~that $G$ is connected and that $1 < \ell < |Z|$. We define a graph $G'$ in the same way as we did in the proof of \cref{thm:mininl-lower-bound}. The only difference is that we do not include the edge $u_2a'_1$. We additionally append a leaf to every vertex of the set $A$. Furthermore, we add the vertices $P = \{p_1, \dots, p_t\}$, $Q = \{q_1, \dots, q_t\}$, and $R = \{r_1, \dots, r_t\}$. We have the following additional edges.

\begin{enumerate}
    \item $p_i z_i \in E(G')$ for all $i \in [t]$,
    \item $p_iq_j, q_ir_j \in E(G')$ for all $i,j \in [t]$,
    \item $r_i a'_1 \in E(G')$ or all $i \in [t]$,
    \item $s_1 r_i \in E(G')$ for all $i \in [t]$,
    \item $Y \cup \{u_2,u_3\} \subseteq N(p_i) \cap N(q_i) \cap N(r_i)$ for all $i \in [t]$.
\end{enumerate}

We first observe that $V(G') \setminus \nZ$ is a vertex cover of $G'$ and that $|V(G') \setminus \nZ| = 8t + 6$. Thus, \cref{lemma:dfs-vc3} implies the following. 

\begin{claim}\label{obs:upper-bound}
    The $\cl$-tree of every LDFS ordering of $G'$ has at most $16t + 12$ internal vertices.
\end{claim}

We now show that the number of internal vertices of the \cl-tree of LDFS orderings of $G'$ is related to the size of a clique containing a vertex of $\nZ$ as long as the LDFS ordering starts with the vertices of $S$.

\begin{claim}\label{lemma:kernel-upper-bound}
    Every LDFS ordering $\sigma$ of $G'$ that starts with the vertices of $S$ visits afterwards a maximal clique of $G$ containing a vertex of $\nZ$. Let $\lambda + 1$ be the size of that clique. Then the \cl-tree of $\sigma$ has $7t + \lambda + 7$ internal vertices.
    Furthermore, for every maximal clique $C$ of $G$ containing a vertex of $\nZ$ there is an LDFS ordering that starts with the vertices of $S$ and then visits the vertices of $C$. 
\end{claim}

\begin{claimproof}
    Similarly as in the proof of \cref{lemma:kernel-lower-bound} of \cref{thm:mininl-lower-bound}, we will not mention when the leaves appended to some vertices are visited. The part where we ``count'' the size of $C$ are now the vertices of $P$ as all other vertices of $V(G') \setminus \nZ$ are ensured to be internal vertices. Visiting the vertices of $C$ ensures that their respective vertices in $P$ are also visited earlier, ensuring that they become internal vertices. 

    Consider an LDFS ordering $\sigma$ of $G'$ starting with the vertices of $S$. Following the arguments in the proof of \cref{lemma:kernel-lower-bound} of \cref{thm:mininl-lower-bound}, we can argue that $\sigma$ starts with $S$, then visits a vertex $v \in \nZ$, then vertices of $Z$ that form a maximal clique $C$ with $v$, then $U$ and then $Y$. 
    
    Afterwards, the only remaining vertices that are adjacent to all vertices of $Y$ are the vertices in $A \cup A' \cup B \cup P \cup Q \cup R$. Among these vertices, the vertices of $P \cup Q \cup R$ are the only vertices that are adjacent both to $u_2$ and $u_3$ while all other vertices are only adjacent to $u_3$ or to none of these vertices. So a vertex of $P \cup Q \cup R$ will be visited next. Vertex $z_\lambda$ is adjacent to $p_\lambda$ and to no other vertex of $P \cup Q \cup R$. All other neighbors of some vertex of $P \cup R \cup Q$ to the right of $z_\lambda$ in $\sigma$ are also adjacent to $p_\lambda$. Thus, $p_\lambda$ is visited next. The only unvisited neighbors of $p_\lambda$ are the vertices in $Q$. So one of them is visited next. W.l.o.g. we may assume that it is $q_1$. Now, either a vertex of $P$ or of $R$ can be visited next. With the same argument as before, we can see that $p_{\lambda - 1}$ must be the next vertex. Repeating this procedure, we see that we traverse the vertices $p_{\lambda}, \dots, p_1$ as well as $\lambda$ vertices of $Q$ ending in some vertex of $Q$. Note that this implies that the vertices $p_1, \dots, p_\lambda$ are internal vertices of the $\cl$-tree. 
    
    The visited neighbors of the remaining vertices of $P$ are the elements of $Y \cup \{u_2, u_3\}$ as well as the visited vertices of $Q$. These vertices are also neighbors of the elements of $R$. Furthermore, the vertices of $R$ are adjacent to $s_1$. Thus, we visit one vertex of $R$ next. As $a'_1$ is not adjacent to any vertex of $P$, the next vertex has to be an element of $Q$. Repeating this argument, we see that we alternate between $Q$ and $R$ until all vertices of $Q$ are visited. This process ends in some vertex of $R$. Afterwards, we visit $a'_1$ and then the remaining vertices of $R$. Finally, we have to traverse the vertices of $A \cup A' \cup B$ in the same way as we did in the proof of \cref{thm:mininl-lower-bound}. This ensures that all vertices of $Z$ are visited before the remaining vertices of $P$ are visited. Therefore, all vertices of $P \setminus \{p_1, \dots, p_\lambda\}$ will be leaves in the \cl-tree. The same holds for all vertices of $\nZ \setminus \{v\}$.

    Finally, let us count the number of internal vertices of the \cl-tree. All vertices of $Z \cup S \cup U \cup Y \cup A \cup A' \cup B \cup Q \cup R \cup \{v\}$ are internal vertices. Furthermore, $\lambda$ vertices of $P$ are internal vertices. Thus, there are $7t + \lambda + 7$ internal vertices.
\end{claimproof}

    Again, we consider copies of $G'$ to overcome the problem with fixing the start vertices of the LDFS ordering. Consider $9t+6$ copies $H_1, \dots, H_{9t+6}$ of $G'$. For each $i \in [9t+6]$, let $S_i$ be the $S$-sets of copy $H_i$. Let $s$ be another vertex that is adjacent to all vertices in all $S_i$-sets. We call the resulting graph~$H$. We claim that $H$ has an LDFS $\cl$-tree with at least $k = (9t+6) \cdot (7t + \ell + 7) + 1$ internal vertices if and only if $G$ has a clique of size $\ell+1$ that contains exactly one vertex of~$\nZ$.

    First, assume that there is such a clique. Then, due to \cref{lemma:kernel-upper-bound}, every $H_i$ has some LDFS ordering $\sigma_i$ starting with the vertices of $S_i$ whose $\cl$-tree has $7t + \ell + 7$ internal vertices. We observe that $(s) \oplus \sigma_1 \oplus \dots \oplus \sigma_{9t+6}$ forms an LDFS ordering of $H$ whose $\cl$-tree has $(9t+6) \cdot (7t + \ell + 7) + 1$ internal vertices. 

    Now assume that such a clique does not exist, i.e., every clique of $G'$ containing a vertex of $\nZ$ contains at most $\ell - 1$ vertices of $Z$. Let $\sigma$ be an LDFS ordering of $H$ and let $\sigma_1, \dots, \sigma_{9t+6}$ be the suborderings of $\sigma$ induced by the vertices of $H_1, \dots, H_{9t+6}$, respectively. It is easy to see that for every $i \in [9t+6]$, the subordering $\sigma_i$ forms an LDFS ordering of $H_i$. Furthermore, all but maybe $\sigma_1$ start with the vertices of set $S_i$. Hence, \cref{lemma:kernel-upper-bound} implies that the \cl-tree of these orderings has at most $7t + \ell + 6$ internal vertices. The $\cl$-tree of $\sigma_1$ has at most $16t + 12$ internal vertices, due to \cref{obs:upper-bound}. Furthermore, $s$ is an internal vertex. Thus, the \cl-tree of $\sigma$ has at most $(9t+5) \cdot (7t + \ell + 6) + (16t + 12) + 1 = (9t+6) \cdot (7t + \ell + 6) + (9t - \ell + 6) + 1 < (9t+6) \cdot (7t + \ell + 7) + 1$ internal vertices. 

    Note that both the vertex cover number of $G'$ and the value $k$ are bounded by $\O(|Z|^2)$. Hence, the given reduction is a polynomial-time \FPT{} reduction from \CVC{} to \MaxInL{} of LDFS parameterized by vertex cover number plus solution size that increases the parameter only polynomially. 
\end{proof}

\subsection[L-trees of Other Searches]{$\cl$-trees of Other Searches}\label{sec:minin}

In this previous section, we have seen that the structure of DFS \cl-trees allows us to transfer the \FPT{} results also to the more sophisticated variant LDFS. In this section, we will study searches whose \cl-trees do not share the structural propertiey of DFS \cl-trees. 
First, we consider the \cl-trees of GS.

\begin{theorem}\label{thm:min-leaf-gs}
    Let $G$ be a connected graph and let $S$ be a set of vertices of $G$. The following statements are equivalent.
    \begin{enumerate}
        \item There is a GS ordering of $G$ whose \cl-tree $T$ has only leaves contained in $S$.
        \item There is a rooted spanning tree $T'$ of $G$ {where all leaves are contained in $S$}.
    \end{enumerate}
    Consequently, there is a rooted GS \cl-tree of $G$ with at least $k$ internal vertices if and only if there is a rooted spanning tree of $G$ with at least $k$ internal vertices.
\end{theorem}

\begin{proof}
    Obviously, Statement~1 implies Statement~2. For the other direction, let $T$ be a rooted spanning tree of $G$ where all the leaves are contained in $S$. Let $\sigma$ be a DFS ordering of $T$ starting in the root of $T$. Note that $\sigma$ is not necessarily a DFS ordering of $G$ but it is a GS ordering of $G$. Let $T'$ be the $\cl$-tree of $\sigma$ with respect to $G$. We claim that all leaves of $T'$ are also leaves in $T$ or equivalently all internal vertices of $T$ are also internal vertices of $T'$. Let $v$ be an internal vertex of $T$. As $\sigma$ is a DFS ordering on $T$, the successor $w$ of $v$ in $\sigma$ is a child of $v$ in $T$, due to \cref{obs:dfs-leaves}. Hence, the edge $vw$ exists in $T'$ and $v$ is not a leaf in $T'$.
\end{proof}

This theorem implies that \MaxInL{} of GS is equivalent to \textsc{Max-Internal Spanning Tree}. For this problem, Li et al.~\cite{li2017deeper} gave a kernel with $2k$ vertices, implying the same for \MaxInL{}.

\begin{corollary}
    \MaxInL{} of GS admits a kernel of $2k$ vertices.
\end{corollary}

Next, we will see that we cannot obtain a similar result for the minimization problem. We will use the following property of $\cl$-trees.

\begin{observation}[Scheffler {\cite[Lemma 20]{scheffler2022recognition}}]\label{obs:children}
    The children of a vertex $v$ in an $\cl$-tree of some GS ordering form an independent set.
\end{observation}

This allows us to prove the following.

\begin{theorem}\label{thm:min-in-bfs}
    \MinInL{} of GS and BFS is \NP-hard for every fixed $k \geq 5$.
\end{theorem}

\begin{proof}
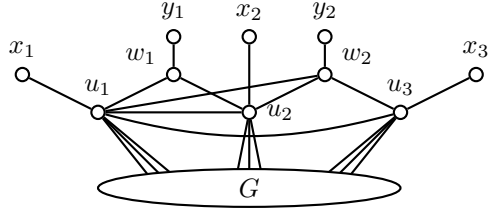
\begin{figure}
    \centering
    \begin{tikzpicture}[scale=1.0]
    \node[vertex,label=90:$u_1$] (u1) at (0,0) {};
    \node[vertex,label=0:$u_2$] (u2) at (2,0) {};
    \node[vertex,label=90:$u_3$] (u3) at (4,0) {};

    \node[vertex,label=175:$w_1$] (w1) at (1,0.5) {};
    \node[vertex,label=5:$w_2$] (w2) at (3,0.5) {};

    \node[vertex,label=90:$x_1$] (x1) at (-1,0.5) {};
    \node[vertex,label=90:$x_2$] (x2) at (2,1) {};
    \node[vertex,label=90:$x_3$] (x3) at (5,0.5) {};

    \node[vertex,label=90:$y_1$] (y1) at (1,1.0) {};
    \node[vertex,label=90:$y_2$] (y2) at (3,1.0) {};

    \draw[normaledge] (u1) -- (w1) -- (u2) -- (w2) -- (u3);
    \draw[normaledge] (x1) -- (u1);
    \draw[normaledge] (x2) -- (u2);
    \draw[normaledge] (x3) -- (u3);
    \draw[normaledge] (w1) -- (y1);
    \draw[normaledge] (w2) -- (y2);

    \draw[normaledge] (u1) -- (u2);
    \draw[normaledge,bend angle=15, bend right] (u1) to (u3);
    \draw[normaledge] (u1) -- (w2);

    \draw[normaledge] (u1) -- (1,-1);
    \draw[normaledge] (u1) -- (0.8,-1);
    \draw[normaledge] (u1) -- (1.2,-1);

    \draw[normaledge] (u2) -- (2,-1);
    \draw[normaledge] (u2) -- (1.8,-1);
    \draw[normaledge] (u2) -- (2.2,-1);

    \draw[normaledge] (u3) -- (3,-1);
    \draw[normaledge] (u3) -- (2.8,-1);
    \draw[normaledge] (u3) -- (3.2,-1);

    \draw[normaledge,fill=white] (2,-1) ellipse (2cm and 0.25cm);

    \node at (2,-1) {$G$};
\end{tikzpicture}
    \caption{Construction of the proof of \cref{thm:min-in-bfs}. Multiple edges to $G$ imply that all edges to vertices in $G$ are present.}
    \label{fig:max-ls-gs}
\end{figure}

    We reduce from \ThreeColor, i.e., we want to decide whether a given graph $G$ is 3-colorable.

    We construct the graph $G'$ from $G$ as follows (see \cref{fig:max-ls-gs}). We add the vertices of the vertex sets $U = \{u_1, u_2, u_3\}$, $W = \{w_1, w_2\}$, $X = \{x_1, x_2, x_3\}$, and $Y = \{y_1, y_2\}$. The vertices of $U$ are universal to $G$. Furthermore, $u_i$ is adjacent to $x_i$. The vertex $w_i$ is adjacent to $u_i$ and to $u_{i+1}$ as well as to $y_i$. Vertex $u_1$ is also adjacent to all other vertices in $U \cup W$. We claim that $G'$ has a GS and BFS ordering whose $\cl$-tree has at most five internal vertices if and only if $G$ is $3$-colorable.

    First, assume that $G$ is $3$-colorable and let $V_1, V_2, V_3$ be the color classes of some $3$-coloring of $G$. We start our BFS ordering $\sigma$ in $u_1$. Then we visit all vertices of $V_1$ and afterwards we visit $w_1$. Next, we visit $u_2$ followed by the vertices of $V_2$ and $w_2$. Then we visit $u_3$ and the vertices of $V_3$. Finally, we visit the vertices of $X$ and $Y$ in the order $x_1, y_1, x_2, y_2, x_3$. This defines a valid BFS ordering of $G'$ as the start vertex $u_1$ is adjacent to all vertices except for the vertices in $X \cup Y$ and the ordering of these vertices follows the ordering of their neighbors.
    
    Let $T$ be the $\cl$-tree of the ordering $\sigma$. We claim that all vertices of $G$ as well as all vertices of $X$ and $Y$ are leaves of $T$. For the vertices of $X$ and $Y$ this is obvious as they are also leaves in $G'$ and are not the root of the tree $T$. So consider a vertex $v$ of $G$ and assume, for contradiction, that this vertex is not a leaf. This implies that $v$ has some child. This child $z$ is either a vertex of $G$ or a vertex of $U$. Obviously, the child cannot be $u_1$ as this vertex is the start vertex of $\sigma$. For every other $u_i$, the vertex $w_{i-1}$ is visited directly before $u_i$. Thus, the parent of $u_i$ is $w_{i-1}$. Therefore, $z$ is contained in $G$. Let $V_j$ be the color class of $v$. As $z$ has to be adjacent to $v$ and has to be visited after $v$ in $\sigma$, it holds that $z$ is in some color class $V_\ell$ with $\ell > j$. However, vertex $u_\ell$ is visited between $v$ and the vertices of $V_\ell$; a contradiction.

    Now assume that $G$ has a GS ordering $\sigma$ whose \cl-tree $T$ has at most five internal vertices. First, observe that the vertices in $U \cup W$ are internal vertices of $T$ since they are cut vertices of $G'$. As $|U \cup W| = 5$, we know that all other vertices have to be leaves in $T$. In particular, this holds for the vertices of $G$. For all $i \in \{1,2,3\}$, let the set $V_i$ contain those vertices of $G$ that are children of $u_i$. As all vertices of $G$ must be leaves and they are only adjacent to vertices of $G$ and of $U$, it holds that $V_1 \cup V_2 \cup V_3 = V(G)$. Due to \cref{obs:children}, the sets $V_i$ induce independent sets in $G$. Therefore, the sets $V_1, V_2, V_3$ induce a coloring of $G$ with at most three colors.
    To show \NP-completeness for every fixed value $k > 5$, we simply add a  path to $u_1$ with exactly $k - 5$ internal vertices.
\end{proof}

This result raises the question whether \MinInL{} of GS and BFS can be solved in polynomial time for fixed values smaller than $5$. Using properties of \cl-trees of GS given in \cite{mfcs}, we can show that at least for these trees it is indeed true. The idea of the algorithm is a complete enumeration of all possible tree structures and orderings of the $k \leq 4$ internal vertices. To decide whether such a choice is valid, one has to decide whether the remaining vertices can be traversed in such a way that they all become leaves. In particular, one has to find out whose children the remaining vertices should become. This question can be reduced to an instance of the \textsc{2-List Coloring} problem, where all lists have size at most~2; a problem that is well-known to be solvable in polynomial time.
We need the following definition and lemmas for our algorithm.

\bigskip
\begin{definition}[Beisegel et al.~\cite{mfcs}]
    Let $T$ be a spanning tree of a graph $G$ rooted in $r \in V(G)$. Let $xy$ be an edge in $E(G) \setminus E(T)$ and let $x'$ and $y'$ be the parents of $x$ and $y$ in $T$, respectively. We call the set $\{x,y,x',y'\}$ a \emph{$U$-bend} of $T$. If $y$ is a descendant of $x'$ but not a descendant of $x$, then we call the triple of vertices $x$, $y$, and $x'$ a \emph{hook configuration}. We call $x$ the \emph{point} and $y$ the \emph{eye} of the hook.
\end{definition}
\begin{lemma}[Beisegel et al.~\cite{mfcs}]\label{lemma:ubend}
    Let $T$ be a spanning tree of a graph $G$ rooted in $r$. Let $xy$ be an edge in $E(G) \setminus E(T)$ and let $x'$ and $y'$ be the parents of $x$ and $y$ in $T$, respectively. If $T$ is an \cl-tree of a GS ordering $\sigma$ starting with $r$, then it either holds that $x' \prec_\sigma x \prec_\sigma y' \prec_\sigma y$ or $y' \prec_\sigma y \prec_\sigma x' \prec_\sigma x$.
\end{lemma}

\begin{lemma}[Beisegel et al.~\cite{mfcs}]\label{lemma:hook}
    Let $x$ and $y$ be part of a hook configuration of $T$ rooted in $r \in V(G)$ with point $x$ and eye $y$. Then for any GS ordering $\sigma$ starting in $r$ with \cl-tree $T$ it holds that $x \prec_\sigma y$.
\end{lemma}

We use these properties of \cl-trees to give the following algorithmic result.

\begin{theorem}\label{thm:Min-Int-GS-poly}
    Given a graph $G$ and a value $k \leq 4$, then \MinInL{} of GS can be solved in polynomial time.
\end{theorem}

\begin{proof}
    The idea of the algorithm is an exhaustive search over all possible sets of internal vertices, every possible tree structure of these internal vertices, and every possible ordering of these vertices in the search ordering. It is clear that there are $\O(n^k)$ of these choices. So, let $I = \{v_1, \dots, v_k\}$ be the choice of the internal vertices with $k \in \{1,2,3,4\}$, let $T'$ be the choice of the subtree induced by $I$, and let $\sigma'$ be the ordering of the vertices in $I$. W.l.o.g.~we assume that $\sigma' = (v_1,\dots,v_k)$. Note that we assume that $T'$ and $\sigma'$ are consistent, i.e., the parent of a vertex in $T'$ is to the left of that vertex in $\sigma'$. Furthermore, we assume that $\sigma'$ fulfills \cref{lemma:hook}, i.e., if there is a hook in $T'$, then its vertices are ordered correctly in $\sigma'$. We now have to decide whether we can add the other vertices as leaves to $T'$. 
    
    We reduce this question to an instance of the \textsc{2-List Coloring} problem, which is known to be solvable in polynomial time.
    To this end, we define for every vertex except $v_1$ a list of colors that encode the possible parents. If $v_j$ is the parent of $v_i \in \{v_2, \dots, v_k\}$, then we set $L(v_i) = \{j\}$. For every vertex $u \in V(G) \setminus I$, we set $L(u) := \{i \in \{1,\dots,k\} \mid uv_i \in E(G)\}$. We apply the following general rules to reduce the size of the sets $L$.

    \begin{enumerate}[(R1)]
        \item If $i,j \in L(u)$ and $v_i$ is the parent of $v_j$ in $T'$, then remove $i$ from $L(u)$.\label{1-appendix}
        \item If $i,j \in L(u)$, $i < j$, and $v_i$ is a descendant of the parent of $v_j$ in $T'$, then remove $i$ from $L(u)$.\label{2-appendix}
        \item If $i,j \in L(u)$, $i < j$, the parent of $v_j$ in $T'$ is $v_\ell$, $v_\ell$ is not the parent of $v_i$ and $\ell < i$, then remove $i$ from $L(u)$.\label{3-appendix}
    \end{enumerate}

    We first show that these reductions are correct, i.e., we do not remove indices of vertices that might be parents of $u$ in some $\cl$-tree of GS following the conditions given by $T'$ and $\sigma'$.

    \begin{claim}\label{claim:parent-appendix}
        Assume there is a GS ordering $\sigma$ of $G$ following $\sigma'$ with an \cl-tree $T$ that contains $T'$ as subtree. Let $v_i$ be the parent of $u$ in $T$. Then $i \in L(u)$.
    \end{claim}

    \begin{claimproof}
        Assume for contradiction that $i \notin L(u)$. As $v_i$ is adjacent to $u$, $i$ has been an element of $L(u)$ before (R\ref{1-appendix}), (R\ref{2-appendix}) and (R\ref{3-appendix}) have been applied. So, first assume that (R\ref{1-appendix}) has removed $i$ from $L(u)$. This implies that $u$ and $v_j$ are both children of $v_i$ in $T$ but $u$ and $v_j$ are adjacent. This contradicts \cref{obs:children}. 
        
        Now assume that (R\ref{2-appendix}) has removed $i$ from $L(u)$. Then, $u$, $v_j$, and the parent of $v_j$ form a hook in $T$. \Cref{lemma:hook} implies that $v_j$ is to the left of $u$ in $\sigma$. Since  $v_i$ is to the left of $v_j$ in $\sigma$, it holds that $v_i \prec_\sigma v_j \prec_\sigma u$. This contradicts the fact that $v_i$ is the parent of $u$ in $T$ as $v_i$ is not the rightmost neighbor of $u$ to the left of $u$ in $\sigma$. 
        
        Finally, assume that (R\ref{3-appendix}) has removed $i$. Then, the situation is as in \cref{lemma:ubend}. Since $v_\ell \prec_\sigma v_i$, it must also hold that $v_j \prec_\sigma v_i$; a contradiction to the fact that $i < j$. 
    \end{claimproof}

    Next, we show that the $L$-sets have a size of at most~2.

    \begin{claim}
        For every $u \in V(G) \setminus \{v_1\}$, it holds that $|L(u)| \leq 2$.
    \end{claim}

    \begin{claimproof}
        Assume that there are $a,b,c \in L(u)$. Then, due to (R\ref{1-appendix}), none of the three vertices $v_a, v_b, v_c$ has a parent in $T'$ that is part of the set $\{v_a, v_b, v_c\}$ and the parents of these vertices (if they exist)  are pairwise different, due to (R\ref{2-appendix}). At least two of the three vertices must have a parent, since at most one of them could be the root of $T$. Thus, there must be at least five vertices in $I$; a contradiction to our assumption that $k \leq 4$.
    \end{claimproof}

    Due to this claim, the $L$-lists form an instance of \textsc{2-List Coloring} for the graph $G - v_1$. We claim that this instance is equivalent to the question  whether there is a proper \cl-tree following the conditions of $T'$ and $\sigma'$. Note that there could be empty lists.

    \begin{claim}
        There is a GS ordering $\sigma$ of $G$ following $\sigma'$ with an \cl-tree $T$ that contains $T'$ as subtree if and only if there is a 2-list coloring for $G - v_1$ following the $L$-lists.
    \end{claim}

    \begin{claimproof}
        First, assume that there is an ordering $\sigma$ with a correct $\cl$-tree $T$. Note that the parent of every vertex in $G - v_1$ is an element of $I$. We define a coloring of $G - v_1$ as follows: For every vertex $u \in G - v_1$ with parent $v_i$, we set $f(u) = i$. We now claim that $f$ is a proper 2-list coloring of $G - v_1$. Due to \cref{claim:parent-appendix}, $i$ was part of the list $L(u)$. So, it only remains to show that $f$ is a proper coloring. As vertices only have the same color if they have the same parent, \cref{obs:children} implies that the color classes are independent sets. Hence, $f$ is a proper coloring following the $L$-lists.
    
        Now assume that there is a proper $k$-coloring $f$ of $G - v_1$ following the $L$-lists. We set $f(v_1) = 0$. We first visit $v_1$ and then all vertices with color $1$ that are not part of $I$ in arbitrary order. We repeat this procedure with all vertices $v_2, \dots, v_k$. We call the resulting ordering $\sigma$. Note that every vertex of $G$ except for $v_1$ has some color from the set $\{1,\dots,k\}$. If it has the color $i$, then it is adjacent to $v_i$ in $G$. Therefore, when the vertex is visited, it has at least one neighbor that was visited before. Thus, the ordering $\sigma$ is a GS ordering of $G$. 
        
        It remains to show that the \cl-tree $T$ of $\sigma$ is correct. It suffices to show that every vertex $u \in G - v_1$ has the parent $v_{f(u)}$. First, consider the case that $u \notin I$. Then all vertices between $v_{f(u)}$ and $u$ in $\sigma$ also have color $f(u)$. As $f$ is a proper coloring, these vertices are not adjacent to $u$ and, thus, $v_{f(u)}$ is the parent of $u$ in $T$.

        So consider the case that $u \in I$, i.e., $u = v_j$ for some $j \in \{1,\dots,k\}$. Assume for contradiction that $w$ is the parent of $u$ in $T$ and $w \neq v_{f(u)}$. As $u$ and $w$ are adjacent, it holds that $f(u) \neq f(w)$. Furthermore, $v_{f(u)}$ is visited before $u$ and, thus, $v_{f(u)} \prec_\sigma w \prec_\sigma u$. 

        First, assume that $w \notin I$. Then $w$ is visited between $v_{f(w)}$ and any other element of $I$, in particular of $u$. This implies that $f(u) < f(w) < j$. If $v_{f(u)}$ is the parent of $v_{f(w)}$, then (R\ref{2-appendix}) would have removed $f(w)$ from $L(w)$. Otherwise, (R\ref{3-appendix}) would have removed $f(w)$ from $L(w)$. So in any case, $f(w)$ would not be part of $L(w)$; a contradiction.

        So finally consider the case that $w \in I$. Thus, $\{u,v_{f(u)},w,v_{f(w)}\} = I$. If $f(u) = 1$, then $v_{f(w)}$ must be a child of $v_{f(u)}$ in $T'$. Thus, the edge $uw$ would induce a hook in $T'$ which implies that $u \prec_{\sigma'} w$; a contradiction. If $f(w) = 1$, then $v_{f(u)}$ must be a child of $v_{f(w)}$ in $T'$ and the hook induced by $uw$ would imply that $w \prec_{\sigma'} f(u)$; again a contradiction.
    \end{claimproof}
    This concludes the proof.
\end{proof}

\section{Number of Leaves as Parameter}\label{sec:leaves}

The para-\NP-hardness of \MinLeafL{} of DFS for $k=1$ follows directly from the fact that DFS can find Hamiltonian paths~\cite{bergougnoux2025parameterized}. The same holds for GS. Even stronger, $\cl$-trees of GS already capture the full power of general spanning trees when one is interested in the minimum number of leaves of the trees, as was shown in \cref{thm:min-leaf-gs}. It is easy to see that for every fixed value~$k \geq 1$, \textsc{Min-Leaf Spanning Tree} is \NP-complete. This implies the following.

\begin{corollary}
    \MinLeafL{} of GS is \NP-complete for all fixed values $k \geq 1$.
\end{corollary}

Using a slight modification of the graph, we can achieve the same result for BFS.

\begin{theorem}\label{thm:min-leaf-bfs}
    \MinLeafL{} of BFS is \NP-complete for every fixed $k \geq 1$.
\end{theorem}

\begin{proof}
    We reduce from \textsc{Hamiltonian Path}. Let $G$ be some graph and let $G'$ be the graph constructed from $G$ by adding a universal vertex $u$ and appending $k$ leaves $w_1, \dots, w_k$ to $u$. 

    We claim that $G'$ has an \cl-tree of BFS with no more than $k$ leaves if and only if $G$ has a Hamiltonian path. First, assume that $G$ has a Hamiltonian path $P$. Then prepending $(w_1, u, w_2, \dots, w_k)$ to an ordering of the vertices of $G$ according to $P$ results in a BFS ordering whose \cl-tree has exactly $k$ leaves.

    So assume that $G'$ has a BFS ordering $\sigma$ whose \cl-tree $T$ has $\leq k$ leaves. If $\sigma$ starts in some leaf $w_i$, then removing $u$ and $w_1, \dots, w_k$ from $T$ must result in a Hamiltonian path of $G$. If $\sigma$ does not start in any leaf $w_i$, then all $w_i$ are leaves in $T$. Removing these leaves from $T$ must result in a Hamiltonian path of $G' \setminus \{w_1, \dots, w_k\}$. As the end vertex of this path has not been a leaf in $T$, this end vertex must be $u$. Hence, this path without $u$ is a Hamiltonian path of $G$.
\end{proof}

\begin{figure}
    \centering
    \begin{tikzpicture}[scale=0.30]

  \node [rectangle,text width=2.5cm,minimum height=0.8cm,
     rounded corners,name = box1] at (4.0,0.7) {};
  
  \filldraw[lipicsLightGray][] (box1.south west)
        [rounded corners=4pt] -- (box1.south east)
        [rounded corners=4pt] -- (box1.north east)
        [rounded corners=4pt] -- (box1.north west)--cycle
        ;

  \node [rectangle,text width=2.5cm,minimum height=0.8cm,
     rounded corners,name = box2] at (14.0,0.7) {};
  
  \filldraw[lipicsLightGray][] (box2.south west)
        [rounded corners=4pt] -- (box2.south east)
        [rounded corners=4pt] -- (box2.north east)
        [rounded corners=4pt] -- (box2.north west)--cycle
        ;

  \node[vertex] (p1) at (0,0) {};
  \node[vertex] (p2) at (2,0) {};
  \node[vertex,fill=black] (p3) at (4,0) {};
  \node[vertex] (p4) at (6,0) {};
  \node[vertex] (p5) at (8,0) {};
  \node[vertex] (p6) at (10,0) {};
  \node[vertex] (p7) at (12,0) {};
  \node[vertex,fill=black] (p8) at (14,0) {};
  \node[vertex] (p9) at (16,0) {};
  \node[vertex] (p10) at (18,0) {};

  \draw[normaledge] (p1) -- (p2);
  \draw[normaledge] (p2) -- (p3);
  \draw[normaledge] (p3) -- (p4);
  \draw[normaledge] (p4) -- (p5);
  \draw[normaledge] (p5) -- (p6);
  \draw[normaledge] (p6) -- (p7);
  \draw[normaledge] (p7) -- (p8);
  \draw[normaledge] (p8) -- (p9);
  \draw[normaledge] (p9) -- (p10);

  \draw[normaledge] (p1) edge[bend left=50] (p4);

  \draw[normaledge] (p2)  edge[bend left] (p4);
  \draw[normaledge] (p2)  edge[bend left=50] (p5);

  \draw[normaledge] (p6) edge[bend left=50] (p9);

  \draw[normaledge] (p7)  edge[bend left] (p9);
  \draw[normaledge] (p7)  edge[bend left=50] (p10);

  \begin{scope}
    \path [scope fading=east] (20,-2) rectangle (27,2);
    
    \node [rectangle,text width=2.5cm,minimum height=0.8cm,
            rounded corners,name = box3] at (24.0,0.7) {};
  
    \filldraw[lipicsLightGray][] (box3.south west)
        [rounded corners=4pt] -- (box3.south east)
        [rounded corners=4pt] -- (box3.north east)
        [rounded corners=4pt] -- (box3.north west)--cycle
        ;

    \node[vertex] (p11) at (20,0) {};
    \node[vertex] (p12) at (22,0) {};
    \node[vertex, fill=black] (p13) at (24,0) {};
    \node[vertex] (p14) at (26,0) {};
    \node[vertex] (p15) at (28,0) {};
    
    \draw[normaledge] (p10) -- (p11);
    \draw[normaledge] (p11) -- (p12);
    \draw[normaledge] (p12) -- (p13);
    \draw[normaledge] (p13) -- (p14);
    \draw[normaledge] (p14) -- (p15);

    \draw[normaledge] (p11) edge[bend left=50] (p14);

    \draw[normaledge] (p12)  edge[bend left] (p14);
    \draw[normaledge] (p12)  edge[bend left=50] (p15);
  \end{scope}

  \node at (29,0.5) {\textbf{\dots}};
    
\end{tikzpicture}

    \caption{The graph has a Hamiltonian path, but the number of leaves in the \cl-tree of every LDFS ordering is $\Omega(n)$. At most one of the filled vertices can be an internal vertex of the \cl-tree. In every box -- besides the one visited first -- both neighbors of the filled vertex have to be visited before the filled vertex.}\label{fig:path-of-bows}
\end{figure}
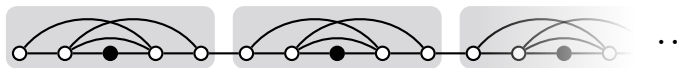

It is not straightforward to adapt these hardness proofs for LDFS. In fact, LDFS is in general not able to find Hamiltonian paths (see \cref{fig:path-of-bows} for an example).
To overcome this issue, we will use an auxiliary graph construction, called a \emph{blow-up}, on which LDFS behaves similar to DFS on the original graph $G$. By showing that the number of leaves in the blow-up corresponds exactly to that of $G$, we can then carry over the reduction via \textsc{Hamiltonian Path}. The blow-up is constructed from $G$ by first subdividing every edge and then taking the line graph of this graph. More formally, the construction is given in the following definition.

\begin{definition}[blow-up]
    Given a graph $G$, the \emph{blow-up} $\B(G)$ of $G$ is defined as follows:
    \begin{itemize}
        \item For each ordered pair $(u,v)$ of adjacent vertices in $G$, the graph $\B(G)$ contains a vertex~$b_u^v$. Note that $(u,v) \neq (v,u)$.
        \item For each $u \in V(G)$, the set $B(u) = \{b_u^v \mid v \in N_G(u)\}$ induces a clique in $\B(G)$.
        \item For each $uv \in E(G)$, there is an edge between $b_u^v$ and $b_v^u$ in $\B(G)$.
    \end{itemize}
\end{definition}

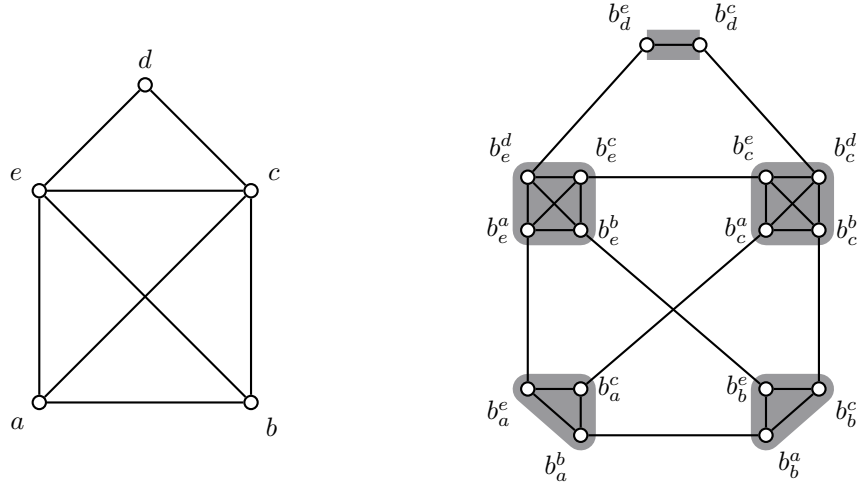
\begin{figure}
  \begin{minipage}{0.49\linewidth}
  \centering
    \begin{tikzpicture}[scale=0.14]

  \node[vertex,label=225:$a$] (a) at (0,0) {};
  \node[vertex,label=315:$b$] (b) at (20,0) {};
  \node[vertex,label=10:$c$] (c) at (20,20) {};
  \node[vertex,label=90:$d$] (d) at (10,30) {};
  \node[vertex,label=170:$e$] (e) at (0,20) {};

  \draw[normaledge] (a) -- (b);
  \draw[normaledge] (a) -- (c);
  \draw[normaledge] (a) -- (e);
  \draw[normaledge] (b) -- (c);
  \draw[normaledge] (b) -- (e);
  \draw[normaledge] (c) -- (d);
  \draw[normaledge] (c) -- (e);
  \draw[normaledge] (d) -- (e);
\end{tikzpicture}

  \end{minipage}
  \begin{minipage}{0.49\linewidth}
  \centering
    \begin{tikzpicture}[scale=0.175]

\tikzset{ wolke/.style={lipicsBulletGray,line width=4 mm, line join=round}}


\coordinate (cab) at (3,0.5);
\coordinate (cac) at (3,4);
\coordinate (cae) at (-1,4);
  
  \filldraw[wolke][] 
        (cab) -- (cac) -- (cae) -- cycle;
       
  \node[vertex,label=250:$b_a^b$] (ab) at (cab) {};
  \node[vertex,label=0:$b_a^c$] (ac) at (cac) {};
  \node[vertex,label=200:$b_a^e$] (ae) at (cae) {};
  
 \coordinate (cba) at (17,0.5);
  \coordinate (cbe) at (17,4);
  \coordinate (cbc) at (21,4);
  
  \filldraw[wolke][] 
        (cba) -- (cbe) -- (cbc) -- cycle;
        
  \node[vertex,label=290:$b_b^a$] (ba) at (cba) {};
  \node[vertex,label=180:$b_b^e$] (be) at (cbe) {};
  \node[vertex,label=340:$b_b^c$] (bc) at (cbc) {};

  \coordinate (cca) at (17,16);
  \coordinate (ccb) at (21,16);
  \coordinate (ccd) at (21,20);
  \coordinate (cce) at (17,20);

  \filldraw[wolke][] 
  (cca) -- (ccb) -- (ccd) -- (cce) --cycle;
  
  \node[vertex,label=180:$b_c^a$] (ca) at (cca) {};
  \node[vertex,label=0:$b_c^b$] (cb) at (ccb) {};
  \node[vertex,label=45:$b_c^d$] (cd) at (ccd) {};
  \node[vertex,label=100:$b_c^e$] (ce) at (cce) {};

  \coordinate (cdc) at (12,30) ;
  \coordinate (cde) at (8,30) ;

  \filldraw[wolke][] 
   (cdc) -- (cde) --cycle;
   
  \node[vertex,label=45:$b_d^c$] (dc) at (cdc) {};
  \node[vertex,label=135:$b_d^e$] (de) at (cde) {};

  \coordinate (cea) at (-1,16);
  \coordinate (ceb) at (3,16);
  \coordinate (cec) at (3,20);
  \coordinate (ced) at (-1,20);
  
  \filldraw[wolke][] 
  (cea) -- (ceb) -- (cec) -- (ced) --cycle;
  
  \node[vertex,label=180:$b_e^a$] (ea) at (cea) {};
  \node[vertex,label=0:$b_e^b$] (eb) at (ceb) {};
  \node[vertex,label=45:$b_e^c$] (ec) at (cec) {};
  \node[vertex,label=135:$b_e^d$] (ed) at (ced) {};

  \draw[normaledge] (ab) -- (ba);
  \draw[normaledge] (ac) -- (ca);
  \draw[normaledge] (ae) -- (ea);
  \draw[normaledge] (be) -- (eb);
  \draw[normaledge] (bc) -- (cb);
  \draw[normaledge] (ec) -- (ce);
  \draw[normaledge] (cd) -- (dc);
  \draw[normaledge] (de) -- (ed);

  \draw[normaledge] (ab) -- (ac);
  \draw[normaledge] (ab) -- (ae);
  \draw[normaledge] (ac) -- (ae);

  \draw[normaledge] (ba) -- (bc);
  \draw[normaledge] (ba) -- (be);
  \draw[normaledge] (bc) -- (be);

  \draw[normaledge] (ca) -- (cb);
  \draw[normaledge] (ca) -- (cd);
  \draw[normaledge] (ca) -- (ce);
  \draw[normaledge] (cb) -- (cd);
  \draw[normaledge] (cb) -- (ce);
  \draw[normaledge] (cd) -- (ce);

  \draw[normaledge] (dc) -- (de);

  \draw[normaledge] (ea) -- (eb);
  \draw[normaledge] (ea) -- (ec);
  \draw[normaledge] (ea) -- (ed);
  \draw[normaledge] (eb) -- (ec);
  \draw[normaledge] (eb) -- (ed);
  \draw[normaledge] (ec) -- (ed);

\end{tikzpicture}

  \end{minipage}
  \caption{A graph $G$ and its blow-up $\B(G)$; gray boxes represent cliques of the blow up.}\label{fig:ex-blow-up}
\end{figure}
An example of a blow-up is given in \cref{fig:ex-blow-up}. Observe that the sets $B(v)$ with $v\in V$ partition $V(\B(G))$. There is a strong relationship between the number of leaves of DFS \cl-trees of a graph $G$ and the number of leaves of LDFS \cl-trees of its blow-up $\B(G)$. The general idea of this relation is the fact that LDFS has to traverse the $B$-sets of the blow-up consecutively and in the same order as the vertices of $G$ are traversed by DFS.

\begin{lemma}\label{lemma:blow-up-dfs}
    A graph $G$ has a DFS \cl-tree with exactly $k$ leaves if and only if $\B(G)$ has an LDFS \cl-tree with exactly $k$ leaves.
\end{lemma}
\begin{proof}
    We have to prove two directions. We first prove that for every DFS ordering of $G$ there is an LDFS ordering of $\B(G)$ with the same number of leaves in the \cl-tree. Afterwards, we prove the reverse direction.
    
    \proofsubparagraph{From DFS to LDFS.}
    First, assume that there is a DFS ordering $\sigma$ of $G$ whose $\cl$-tree $T$ has exactly $k$ leaves. We construct an ordering $\beta$ of $\B(G)$ as follows. Let $u$ be a vertex of $G$.

    \begin{itemize}
        \item Let $A(u)$ be the set of vertices $b_u^v$ where $v$ is to the left of $u$ in $\sigma$, let $C(u)$ be the set of vertices $b_u^v$ where $v$ is a child
        of $u$ in $T$ and $D(u) := B(u) \setminus (A(u) \cup C(u))$.
        \item Let $\tau_A(u)$ be an ordering of $A(u)$ such that $b_u^{v_1} \prec_{\tau_A} b_u^{v_2}$ if and only if $v_2 \prec_\sigma v_1$ for all $v_1, v_2 \in A(u)$.
        \item Let $\tau_C(u)$ be an ordering of $C(u)$ such that $b_u^{v_1} \prec_{\tau_A} b_u^{v_2}$ if and only if $v_2 \prec_\sigma v_1$ for all $v_1, v_2 \in C(u)$.
        \item Let $\tau_D(u)$ be an arbitrary ordering of $D(u)$.
        \item Now define $\beta$ by replacing $u$ in $\sigma$ by $\tau_A(u) \oplus \tau_D(u) \oplus \tau_C(u)$ (here $\oplus$ stands for the concatenation operator).
    \end{itemize}

    \begin{claim}
        The ordering $\beta$ is an LDFS ordering of $\B(G)$.
    \end{claim}
    \begin{claimproof}
    Let $x$, $y$, and $z$ be three vertices such that $x \prec_\beta y \prec_\beta z$, $xz \in E(\B(G))$ but $xy \notin E(\B(G))$. Due to the 4-point condition of LDFS (\cref{lemma:4point-ldfs}), it is sufficient to show that there is a vertex $w$ with $x \prec_\beta w \prec_\beta y$ such that $wy \in E(\B(G))$ but $wz \notin E(\B(G))$.
    
    Let $x \in B(u_1)$, $y \in B(u_2)$ and $z \in B(u_3)$. Since $x$ and $y$ are not adjacent, it holds that $u_1 \neq u_2$. By construction, the vertices of $B(u_1)$ are consecutive in $\beta$. Thus, it follows that $u_1 \neq u_3$. This implies that $x = b_{u_1}^{u_3}$ and $z = b_{u_3}^{u_1}$. Observe that this implies that all neighbors of $z$ to the left of $z$ in $\beta$ are either equal to $x$ or elements of $B(u_3)$. 

    Now assume first that $y$ is not the leftmost vertex of $B(u_2)$ in $\beta$. Then, the leftmost vertex $y'$ of $B(u_2)$ is adjacent to $y$ but not to $z$, due to the observation about the neighbors of $z$ above. Furthermore, by construction, $y'$ is placed between $x$ and $y$ and $y'$ is a possible choice for the vertex $w$ showing that the 4-point condition is fulfilled in this case.

    Thus, we may assume that $y$ is the first vertex of $B(u_2)$ in $\beta$. As $u_2$ is not the first vertex of $\sigma$, the set $A(u_2)$ has not been empty and, thus, $y$ is part of $A(u_2)$. Let $y = b_{u_2}^{t}$. By construction, $t$ is the rightmost neighbor of $u_2$ to the left of $u_2$ in $\sigma$, i.e., $t$ is the parent of $u_2$ in $T$ implying that $b_t^{u_2} \in C(t)$. Due to the 4-point condition of DFS (\cref{lemma:4point-dfs}), $t$ is either placed between $u_1$ and $u_2$ in $\sigma$ or $t$ is equal to $u_1$. In the first case, $b_t^{u_2}$ is between $x$ and $y$ in $\beta$ and can take the role of $w$. Thus, we may assume that $t = u_1$. Recall that then $x = b_{u_1}^{u_3}$, $y = b_{u_2}^{u_1}$ and $b_{u_1}^{u_2} \in C(u_1)$. If $u_3$ is not a child of $u_1$ in $T$, then $b_{u_1}^{u_3} \in D(u_1)$ and thus $b_{u_1}^{u_3}$ is to the left of $b_{u_1}^{u_2}$ in $\beta$. Therefore, $b_{u_1}^{u_2}$ can be chosen as $w$. Otherwise, both $b_{u_1}^{u_2}$ and $b_{u_1}^{u_3}$ are part of $C(u_1)$. Since $u_2 \prec_\sigma u_3$, it holds that $x = b_{u_1}^{u_3} \prec_\beta b_{u_1}^{u_2}$ and, again, $b_{u_1}^{u_2}$ can be used as~$w$.
    \end{claimproof}

    \begin{claim}
        The \cl-tree of $\beta$ has exactly $k$ leaves.
    \end{claim}

    \begin{claimproof}
    First observe that every set $B(v)$ contains at most one leaf of $T_\beta$ and if there is one, then this leaf must be the last vertex of $B(v)$. 
    Assume that vertex $u$ is a leaf in $T$. Then all neighbors of $u$ in $G$ are to the left of $u$ in~$\sigma$, due to \cref{obs:dfs-leaves}. Hence, for each $x \in B(u)$, it holds that the single neighbor of $x$ outside of $B(u)$ is to the left of $x$ in $\beta$. In particular, this means that all neighbors of the rightmost vertex of $B(u)$ are to the left of this vertex in $\beta$. Therefore, this vertex is a leaf in $T_\beta$. 

    If $u$ is not a leaf in $T$, then the successor $v$ of $u$ in $\sigma$ is the leftmost child of $u$ in $T$ while $u$ is the rightmost neighbor of $v$ to the left of $v$. This implies that $b_u^v$ is the rightmost vertex of $B(u)$ and $b_v^u$ is the leftmost vertex of $B(v)$ in $\beta$. Therefore, $b_v^u$ is the successor of $b_u^v$ in $\beta$ and thus $b_u^v$ is not a leaf in $T_\beta$. Due to the observation above, no vertex of $B(u)$ is a leaf of $T_\beta$, implying that the number of leaves of $T$ and $T_\beta$ are identical.
    \end{claimproof}
    
    \proofsubparagraph{From LDFS to DFS.}

    Let $\beta$ be an LDFS ordering of $\B(G)$. We first consider the case where for each $v \in V(G)$ the vertices of the clique $B(v)$ appear consecutively in $\beta$. Let $\sigma$ be the vertex ordering of $G$ that is constructed by replacing, for every $v \in V(G)$, the ordering of $B(v)$ in $\beta$ by $v$. 
    
    \begin{claim}
        The ordering $\sigma$ is a DFS ordering of $G$
    \end{claim}
    
    \begin{claimproof}
        Let $x$, $y$, and $z$ be three vertices of $G$ with $x \prec_\sigma y \prec_\sigma z$, $xz \in E(G)$, and $xy \notin E(G)$. Then there are vertices $x' \in B(x)$, $y' \in B(y)$, and $z' \in B(z)$ such that $x' \prec_\beta y' \prec_\beta z'$, $x'z' \in E(\B(G))$ and $x'y' \notin E(\B(G))$. Since no vertex of $B(x)$ is adjacent to any vertex in $B(y)$, we can choose $y'$ as the leftmost vertex of $B(y)$. Since $\beta$ is an LDFS ordering, there is a vertex $w' \in V(\B(G))$ with $x' \prec_\beta w' \prec_\beta y'$ and $w'y' \in E(\B(G))$. Note that $w'$ cannot be an element of $B(x)$ as no vertex of $B(x)$ is adjacent to $y'$. Furthermore, $w'$ cannot be in $B(y)$ since $y'$ is the leftmost vertex of $B(y')$ in $\beta$. Thus, there is a vertex $w \in V(G) \setminus \{x,y,z\}$ with $w' \in B(w)$. It holds that $x \prec_\sigma w \prec_\sigma y$ and $wy \in E(G)$. Therefore, $\sigma$ is a DFS ordering of $G$.
    \end{claimproof}

    \begin{claim}
        The \cl-trees of $\beta$ and $\sigma$ have the same number of leaves.
    \end{claim}

    \begin{claimproof}
    Let $v' \in B(v)$ be a leaf of the $\cl$-tree of $\beta$. Due to \cref{obs:dfs-leaves}, vertex $v'$ is the last vertex of $B(v)$ in $\beta$. Furthermore, its unique neighbor $x' \in B(x)$ in $V(\B(G)) \setminus B(v)$ is to the left of all vertices of $B(v)$ in $\beta$ since $B(v)$ appears consecutively in $\beta$. This implies that for all other vertices $v''$ of $B(v)$ their unique neighbor outside of $B(v)$ must also be to the left of $B(v)$ in $\beta$ (otherwise $x'$, $v''$ and $v'$ would contradict the 4-point condition of LDFS). Therefore, all neighbors of $v$ in $G$ are to the left of $v$ in $\sigma$ and $v$ is a leaf in the $\cl$-tree of $\sigma$. Thus, the number of leaves of the $\cl$-tree of $\sigma$ is at least the number of leaves of the $\cl$-tree of~$\beta$.

    Now, let $v$ be a leaf in $\sigma$. Then the successor $w$ of $v$ in $\sigma$ is not a neighbor of $v$. This implies that there is no edge between $B(v)$ and $B(w)$. In particular, the last vertex of $B(v)$ in $\beta$ is not adjacent to its successor, i.e., the first vertex of $B(w)$ in $\beta$. Due to \cref{obs:dfs-leaves}, the last vertex of $B(v)$ is a leaf in $\beta$ and, thus, the number of leaves of the $\cl$-tree of $\beta$ is equal to the number of leaves of the $\cl$-tree of $\sigma$.
    \end{claimproof}

    It remains to consider the case that not all cliques $B(v)$ appear consecutively in $\beta$. If LDFS has visited two vertices of the same clique $B(v)$, then it must visit all remaining vertices of $B(v)$ since there is no other clique in $\B(G)$ that contains both of these vertices. Whenever we enter a clique $B(v)$ (i.e., there is a vertex visited before), all remaining neighbors of that entry vertex are elements of $B(v)$. Therefore, we have to visit a vertex of $B(v)$ afterwards. This implies that all cliques $B(v)$ appear consecutively in $\beta$ with one possible exception: the clique that contains the first vertex of $\beta$. It is possible that we start in a vertex $s$ of a non-trivial clique $B(v)$ and then directly go to a vertex of $B(w)$. If $B(w)$ only contains a single vertex, then we can just move this vertex to the end of the ordering. This results in an LDFS ordering with the same number of leaves. Otherwise, it is easy to see that there is an LDFS ordering $\beta'$ where all cliques $B(u)$ appear consecutively and that follows $\beta$ with the only exception that vertex $s$ is placed at some position in $B(v)$.

    \begin{claim}
        The \cl-trees of $\beta$ and $\beta'$ have the same number of leaves.
    \end{claim}

    \begin{claimproof}
    Let $T$ be the \cl-tree of $\beta$ and $T'$ be the \cl-tree of $\beta'$. 
    Let $x$ be a leaf of $T$. Then, due to the observations above, $x$ is the last vertex of its clique $B(u)$ and the next vertex in $\beta$ is not adjacent to $x$. If $x$ is not a leaf in $\beta'$, then the successor of $x$ in $\beta'$ is $s$ and $s$ is adjacent to $x$. As the only neighbor of $s$ outside of $B(v)$ is the start vertex of $\beta'$, it holds that both $x$ and $s$ are part of $B(u)$, meaning that $s$ is the last vertex of $B(u)$ in $\beta'$. Therefore, all neighbors of $s$ are to the left of $s$ in $\beta'$, which implies that $s$ is a leaf in $T'$. Therefore, the number of leaves of $T'$ is at least the number of leaves of~$T$.

    For the other direction, we observe that $s$ is the last vertex of its clique $B(v)$ in $\beta'$ if and only if the neighbors of all other vertices of $B(v)$ that are not part of $B(v)$ are to the left of $B(v)$ in $\beta$. Now assume that there is a vertex $x$ that is a leaf in $T'$ but not in $T$. If $x = s$, then, due to the observation above, the vertex before $s$ in $\beta'$ is a leaf in $T$ and is an internal vertex of $T'$. If $x \neq s$, then in $\beta'$ vertex $s$ is placed between $x$ and its successor in $\beta$. The observation implies that $s$ and $x$ are in $B(v)$, contradicting the fact that $x$ is a leaf in~$T'$.\end{claimproof}
    This finalizes the proof.
\end{proof}

This lemma implies that the problems \MinLeafL{} and \MaxLeafL{} of LDFS are at least as hard as the same problems of DFS. In~\cite{bergougnoux2025parameterized}, it was shown that \MinLeafL{} of DFS is \NP-complete for all fixed $k \geq 1$ while \MaxLeafL{} of DFS is \W-hard. So the same holds for LDFS.

\begin{theorem}~
\MinLeafL{} of LDFS is \NP-complete for every fixed $k \geq 1$ and
\MaxLeafL{} of LDFS is \W-hard.
\end{theorem}

\section{Further Research}

\cref{tab:results-tree} summarizes the solved and open questions. In particular, the parameterized complexity of \MaxLeafL{} is still largely unknown for all searches apart from (L)DFS, while for these two, no \XP-algorithm is known to complement the \W-hardness~results.

When comparing the problems studied here with MLST and MIST, we can make some immediate observations. Any \cl-tree is also a spanning tree, so for any~$k$ every solution to \MaxLeaf{} or \MinInL{} yields a lower bound for MLST. Similarly, for any~$k$, every solution to \MinLeaf{} or \MaxInL{} yields a lower bound for MIST. The question remains as to how large the gap between these problems can be. \cref{thm:min-leaf-gs} states that a spanning tree has at most $k$ leaves if and only if there is an \cl-tree of GS that has at most $k$ leaves. For complete graphs, any \cl-tree has at most one leaf, while a spanning tree can have a maximum of $n-1$ leaves. In the graph of~\cref{fig:path-of-bows}, the Hamiltonian path has one leaf, while any \cl-tree computed with LDFS has $\operatorname{\Omega}(n)$ leaves. While these observations shed some light on the gap between spanning trees and \cl-trees, this is far from the complete picture. Regarding \MaxInL{} for BFS and DFS, we lack any estimates on the gap to MIST.

\bibliography{many-leaves}

\end{document}